\documentclass[default,iicol]{sn-jnl}

\usepackage[utf8]{inputenc}
\usepackage{centernot}
\usepackage{amssymb}
\usepackage{amsmath}
\usepackage{amsthm}
\usepackage[normalem]{ulem}
\usepackage{centernot}
\usepackage{xspace}
\usepackage{stmaryrd}
\usepackage{mathrsfs}
\usepackage{proof}
\usepackage{bm}
\usepackage{multirow}
\usepackage{array}
\usepackage{mathtools}
\usepackage{enumitem}
\usepackage{cleveref}
\usepackage{microtype}
\usepackage[justification=centering]{caption}

\newcommand*{\doi}[1]{\href{https://doi.org/#1}{https://doi.org/#1}}

\usepackage[location=inline]{moveproofs}
\newcommand{\submission}[1]{}
\newcommand{\report}[1]{#1}

\DeclareFontFamily{U}{mathx}{\hyphenchar\font45}
\DeclareFontShape{U}{mathx}{m}{n}{
      <5> <6> <7> <8> <9> <10>
      <10.95> <12> <14.4> <17.28> <20.74> <24.88>
      mathx10
      }{}
\DeclareSymbolFont{mathx}{U}{mathx}{m}{n}
\DeclareFontSubstitution{U}{mathx}{m}{n}
\DeclareMathAccent{\widecheck}{0}{mathx}{"71}

\hypersetup{
  pdftitle={A Calculus for Modular Loop Acceleration and Non-Termination Proofs},
  pdfauthor={Florian Frohn and Carsten Fuhs}
}

\renewcommand{\epsilon}{\varepsilon}
\let\oldphi\phi
\let\oldvarphi\varphi
\renewcommand{\phi}{\oldvarphi}
\renewcommand{\varphi}{\oldphi}
\renewcommand{\hat}[1]{\widehat{#1}}
\renewcommand{\check}[1]{\widecheck{#1}}

\newcommand{\ZZ}{\mathbb{Z}}
\newcommand{\AAA}{\mathscr{C}}
\newcommand{\QQ}{\mathbb{Q}}
\newcommand{\TT}{\mathcal{T}}
\newcommand{\RR}{\mathbb{R}}
\newcommand{\NN}{\mathbb{N}}
\newcommand{\VV}{\mathcal{V}}
\newcommand{\charfun}[1]{I_{#1}}
\newcommand{\Def}{\mathrel{\mathop:}=}
\newcommand{\assign}{\leftarrow}

\newcommand{\mDo}{\mathbf{do}}
\newcommand{\mWhile}[2]{\mathbf{while}\ #1\ \mDo\ #2}
\newcommand{\relmiddle}[1]{\mathrel{}\middle#1\mathrel{}}
\newcommand{\MLoop}{\mathit{Loop}}
\newcommand{\Prop}[1]{\mathit{FO_{QF}}(#1)}
\newcommand{\accelerate}{\mathit{accel}}
\newcommand{\nonterm}{\mathit{nt}}
\newcommand{\closure}{\mathit{closure}}
\newcommand{\mf}{\mathit{mf}}
\newcommand{\expr}{e}
\newcommand{\mat}[1]{\left(\begin{smallmatrix} #1 \end{smallmatrix}\right)}
\newcommand{\tool}[1]{\textsf{#1}}
\newcommand{\loat}{\tool{LoAT}\xspace}
\newcommand{\prob}[4]{\left\llbracket #1 \relmiddle{\vert} #2 \relmiddle{\vert} #3 \right\rrbracket_{#4}}
\newcommand{\ntprob}[4]{\left\| #1 \relmiddle{\vert} #2 \relmiddle{\vert} #3 \right\|_{#4}}
\newcommand{\nt}[2]{#1 \longrightarrow_{#2}^ \infty \bot}
\newcommand{\leadstont}{\leadsto_{\mathit{nt}}}
\newcommand{\CC}{{\sf C\nolinebreak\hspace{-.1em}\raisebox{.4ex}{\tiny\bf +}\nolinebreak\hspace{-.10em}\raisebox{.4ex}{\tiny\bf +}}}

\DeclareMathOperator{\dom}{dom}

\crefname{chapter}{chapter}{chapters}%
\crefname{enumi}{item}{items}%
\crefname{footnote}{footnote}{footnotes}%
\crefname{table}{table}{tables}%
\crefname{proposition}{proposition}{propositions}%
\crefname{result}{result}{results}%
\crefname{remark}{remark}{remarks}%
\crefname{note}{note}{notes}%

\title[ ]{A Calculus for Modular Loop Acceleration and Non-Termination Proofs}
\author[1]{Florian Frohn}
\author[2]{Carsten Fuhs}
\affil[1]{Max Planck Institute for Informatics, Saarland Informatics Campus, Saarbr\"ucken, Germany}
\affil[1]{LuFG Informatik 2, RWTH Aachen University, Aachen, Germany}
\affil[2]{Department of Computer Science and Information Systems, Birkbeck, University of London, London, UK}

\theoremstyle{thmstyleone}%
\newtheorem{theorem}{Theorem}%
\newtheorem{lemma}{Lemma}%

\theoremstyle{thmstyletwo}%
\newtheorem{examplex}{Example}%
\newenvironment{example}
  {\pushQED{\qed}\examplex}
  {\popQED\endexamplex}
\newtheorem{remark}{Remark}%

\theoremstyle{thmstylethree}%
\newtheorem{definition}{Definition}%

\begin{document}

\abstract{
  Loop acceleration can be used to prove safety, reachability, runtime bounds, and (non-)termination of programs.
  To this end, a variety of acceleration techniques has been proposed.
  However, so far all of them have been monolithic, i.e., a single loop could not be accelerated using a \emph{combination of several different} acceleration techniques.
  In contrast, we present a calculus that allows for combining acceleration techniques in a modular way and we show how to integrate many existing acceleration techniques into our calculus.
  Moreover, we propose two novel acceleration techniques that can be incorporated into our calculus seamlessly.
  Some of these acceleration techniques apply only to non-terminating loops.
  Thus, combining them with our novel calculus results in a new, modular approach for proving non-termination.
  An empirical evaluation demonstrates the applicability of our approach, both for loop acceleration and for proving non-termination.
}

\maketitle

\section{Introduction}
\label{sec:intro}

In the last years, loop acceleration techniques have successfully been used to build static analyses for programs operating on integers \cite{underapprox15,fmcad19,journal,iosif17,Bozga14,fast}.
Essentially, such techniques extract a quantifier-free first-order formula $\psi$ from a single-path loop $\TT$, i.e., a loop without branching in its body, such that $\psi$ under-approximates (or is equivalent to) $\TT$.
More specifically, each model of the resulting formula $\psi$ corresponds to an execution of $\TT$ (and vice versa).
By integrating such techniques into a suitable program-analysis framework \cite{journal,fmcad19,iosif12,iosif17,FlatFramework}, whole programs can be transformed into first-order formulas which can then be analyzed by off-the-shelf solvers.
Applications include proving safety \cite{iosif12} or reachability \cite{iosif12,underapprox15}, deducing bounds on the runtime complexity \cite{journal}, and proving (non-)termination \cite{fmcad19,Bozga14}.

However, existing acceleration techniques apply only if certain prerequisites are in place.
So the power of static analyses built upon loop acceleration depends on the applicability of the underlying acceleration technique.

In this paper, we introduce a calculus which allows for combining several acceleration techniques modularly in order to accelerate a single loop.
This not only allows for modular combinations of standalone techniques, but it also enables interactions between different acceleration techniques, allowing them to obtain better results together.
Consequently, our calculus can handle classes of loops where all standalone techniques fail.
Moreover, we present two novel acceleration techniques and integrate them into our calculus.

One important application of loop acceleration is proving non-termination.
As already observed in \cite{fmcad19}, certain properties of loops -- in particular monotonicity of (parts of) the loop condition w.r.t.\ the loop body -- are crucial for both loop acceleration and proving non-termination.
In \cite{fmcad19}, this observation has been exploited to develop a technique for deducing invariants that are helpful to deal with non-terminating as well as terminating loops: For the former, they help to prove non-termination, and for the latter, they help to accelerate the loop.

In this paper, we take the next step by also unifying the actual techniques that are used for loop acceleration and for proving non-termination.
To this end, we identify loop acceleration techniques that, if applied in isolation, give rise to non-termination proofs.
Furthermore, we show that the combination of such non-termination techniques via our novel calculus for loop acceleration gives rise to non-termination proofs, too.
In this way, we obtain a modular framework for combining several different non-termination techniques in order to prove non-termination of a single loop.

In the following, we introduce preliminaries in \Cref{sec:preliminaries}.
Then, we discuss existing acceleration techniques in \Cref{sec:monotonic}.
In \Cref{sec:integration}, we present our calculus to combine acceleration techniques and show how existing acceleration techniques can be integrated into our framework.
\Cref{sec:conditional} lifts existing acceleration techniques to \emph{conditional} acceleration techniques, which provides additional power in the context of our framework by enabling interactions between different acceleration techniques.
Next, we present two novel acceleration techniques and incorporate them into our calculus in \Cref{sec:accel}.
Then we adapt our calculus and certain acceleration techniques for proving non-termination in \Cref{sec:nonterm}.
After discussing related work in \Cref{sec:related}, we demonstrate the applicability of our approach via an empirical evaluation in \Cref{sec:experiments} and conclude in \Cref{sec:conclusion}.

A conference version of this paper was published in \cite{conference}.
The present paper provides the following additional contributions:
\begin{itemize}
\item We present formal correctness proofs for all of our results, which were omitted in \cite{conference} for reasons of space.
\item We present an improved version of the loop acceleration technique from \cite[Thm.\ 3.8]{journal} and \cite[Thm.\ 7]{ijcar16} that yields simpler formulas.
\item We prove an informal statement from \cite{conference} on using arbitrary existing acceleration techniques in our setting, resulting in the novel \Cref{thm:accel-as-cond-accel}.
\item The adaptation of our calculus and of certain acceleration techniques for proving non-termination (\Cref{sec:nonterm}) is completely new.
\item We extend the empirical evaluation from \cite{conference} with extensive experiments comparing the adaptation of our calculus for proving non-termination with other state-of-the-art tools for proving non-termination (\Cref{sec:eval-nonterm}).
\end{itemize}

\section{Preliminaries}
\label{sec:preliminaries}

We use the notation $\vec{x}$, $\vec{y}$, $\vec{z}$, ... for vectors.
Let $\AAA(\vec{z})$ be the set of \emph{closed-form expressions} over the variables $\vec{z}$.
So $\AAA(\vec{z})$ may, for example, be defined to be the smallest set containing all expressions built from $\vec{z}$, integer constants, and binary function symbols $\{{+},{-},{\cdot},{/},{\exp}\}$ for addition, subtraction, multiplication, division, and exponentiation.
However, there is no widely accepted definition of ``closed forms'', and the results of the current paper are independent of the precise definition of $\AAA(\vec{z})$ (which may use other function symbols).
Thus, we leave its definition open to avoid restricting our results unnecessarily.
We consider loops of the form
\begin{equation}
  \label{loop}\tag{\ensuremath{\TT_{loop}}}
  \mWhile{\phi}{\vec{x} \assign \vec{a}}
\end{equation}
where $\vec{x}$ is a vector of $d$ pairwise different variables that range over the integers, the loop condition $\phi \in \Prop{\AAA(\vec{x})}$ (which we also call \emph{guard}) is a finite quantifier-free first-order formula over the atoms $\{p>0 \mid p \in \AAA(\vec{x})\}$, and $\vec{a} \in \AAA(\vec{x})^d$ such that the function\footnote{i.e., the (anonymous) function that maps $\vec{x}$ to $\vec{a}$} $\vec{x} \mapsto \vec{a}$ maps integers to integers.
$\MLoop$ denotes the set of all such loops.

We identify \ref{loop} and the pair $\langle \phi, \vec{a} \rangle$.
Moreover, we identify $\vec{a}$ and the function $\vec{x} \mapsto \vec{a}$, where we sometimes write $\vec{a}(\vec{x})$ to make the variables $\vec{x}$ explicit.
We use the same convention for other (vectors of) expressions.
Similarly, we identify the formula $\phi(\vec{x})$ (or just $\phi$) with the predicate $\vec{x} \mapsto \phi$.
We use the standard integer-arithmetic semantics for the symbols occurring in formulas.

Throughout this paper, let $n$ be a designated variable ranging over $\NN = \{0,1,2,\ldots\}$ and let:
\[
  \vec{a} \Def \mat{a_1\\\ldots\\a_d} \hspace{0.83em} \vec{x} \Def \mat{x_1\\\ldots\\x_d} \hspace{0.83em} \vec{x}' \Def \mat{x'_1\\\ldots\\x'_d} \hspace{0.83em} \vec{y} \Def \mat{\vec{x}\\n\\\vec{x}'}
\]
Intuitively, the variable $n$ represents the number of loop iterations and $\vec{x}'$ corresponds to the values of the program variables $\vec{x}$ after $n$ iterations.

$\ref{loop}$ induces a relation ${\longrightarrow_{\ref{loop}}}$ on $\ZZ^d$:
\[
  \phi(\vec{x}) \land \vec{x}' = \vec{a}(\vec{x}) \iff \vec{x} \longrightarrow_{\ref{loop}} \vec{x}'
\]

\section{Existing Acceleration Techniques}
\label{sec:monotonic}

In the following (up to and including \Cref{sec:accel}), our goal is to \emph{accelerate} \ref{loop}, i.e., to find a formula $\psi \in \Prop{\AAA(\vec{y})}$ such that
\begin{equation}
  \label{eq:equiv}\tag{equiv}
  \psi \iff \vec{x} \longrightarrow_{\ref{loop}}^n \vec{x}' \qquad \text{for all } n>0.
\end{equation}
To see why we use $\AAA(\vec{y})$ instead of, e.g., polynomials, consider the loop
\begin{equation}
  \label{loop:exp}
  \tag{\ensuremath{\TT_{exp}}}
  \mWhile{x_1>0}{\mat{x_1\\x_2} \assign \mat{x_1-1\\2 \cdot x_2}.}
\end{equation}
Here, an acceleration technique synthesizes, e.g., the formula
\begin{equation}
  \label{psi:exp}
  \tag{\ensuremath{\psi_{exp}}}
  \mat{x_1'\\x_2'} = \mat{x_1-n\\2^n \cdot x_2} \land x_1 - n + 1 > 0,
\end{equation}
where $\mat{x_1-n\\2^n \cdot x_2}$ is equivalent to the value of $\mat{x_1\\x_2}$ after $n$ iterations, and the inequation $x_1 - n + 1 > 0$ ensures that \ref{loop:exp} can be executed at least $n$ times.
Clearly, the growth of $x_2$ cannot be captured by a polynomial, i.e., even the behavior of quite simple loops is beyond the expressiveness of polynomial arithmetic.

In practice, one can restrict our approach to weaker classes of expressions to ease automation, but the presented results are independent of such considerations.

Some acceleration techniques cannot guarantee \eqref{eq:equiv}, but the resulting formula is an under-approximation of \ref{loop}, i.e., we have
\begin{equation}
  \label{eq:approx}\tag{approx}
  \psi \implies \vec{x} \longrightarrow_{\ref{loop}}^n \vec{x}' \qquad \text{for all } n>0.
\end{equation}
If \eqref{eq:equiv} holds, then  $\psi$ is \emph{equivalent} to \ref{loop}.
Similarly, if \eqref{eq:approx} holds, then $\psi$ \emph{approximates} \ref{loop}.\footnote{While there are also over-approximating acceleration techniques (see \Cref{sec:related-accel}), in this paper we are interested only in under-approximations.}

\begin{definition}[Acceleration Technique]
  \label{def:accel}
  An \emph{acceleration technique} is a partial function
  \[
    \accelerate: \MLoop \rightharpoonup \Prop{\AAA(\vec{y})}.
  \]
  It is \emph{sound} if the formula $\accelerate(\TT)$ approximates $\TT$ for all $\TT \in \dom(\accelerate)$.
  It is \emph{exact} if the formula $\accelerate(\TT)$ is equivalent to $\TT$ for all $\TT \in \dom(\accelerate)$.
\end{definition}

We now recall several existing acceleration techniques.
In \Cref{sec:integration} we will see how these techniques can be combined in a modular way.
All of them first compute a \emph{closed form} $\vec{c} \in \AAA(\vec{x},n)^d$ for the values of the program variables after $n$ iterations.
\begin{definition}[Closed Form]
  \label{def:closed}
  We call $\vec{c} \in \AAA(\vec{x},n)^d$ a \emph{closed form} of \ref{loop} if
  \[
    \forall \vec{x} \in \ZZ^d, n \in \NN.\ \vec{c} = \vec{a}^n(\vec{x}).
  \]
  Here, $\vec{a}^n$ is the $n$-fold application of $\vec{a}$, i.e., $\vec{a}^0(\vec{x}) = \vec{x}$ and $\vec{a}^{n+1}(\vec{x}) = \vec{a}(\vec{a}^n(\vec{x}))$.
\end{definition}

To find closed forms, one tries to solve the system of recurrence equations $\vec{x}^{(n)} = \vec{a}(\vec{x}^{(n-1)})$ with the initial condition $\vec{x}^{(0)} = \vec{x}$.
In the sequel, we assume that we can represent $\vec{a}^n(\vec{x})$ in closed form.
Note that one can always do so if $\vec{a}(\vec{x}) = A\vec{x} + \vec{b}$ with $A \in \ZZ^{d \times d}$ and $\vec{b} \in \ZZ^d$, i.e., if $\vec{a}$ is linear.
To this end, one considers the matrix $B \Def \mat{A & \vec{b} \\ \vec{0}^T & 1}$ and computes its Jordan normal form $J$ where $B = T^{-1}JT$ and $J$ is a block diagonal matrix (which has complex entries if $B$ has complex eigenvalues).
Then the closed form for $J^n$ can be given directly (see, e.g., \cite{Ouaknine15}), and $\vec{a}^n(\vec{x})$ is equal to the first $d$ components of $T^{-1}J^nT\mat{\vec{x}\\1}$.
Moreover, one can compute a closed form if $\vec{a} = \mat{c_1 \cdot x_1 + p_1\\\ldots\\c_d \cdot x_d + p_d}$ where $c_i \in \NN$ and each $p_i$ is a polynomial over $x_1,\ldots,x_{i-1}$ \cite{polyloopsLPAR20,polyloopsSAS20}.

\subsection{Acceleration via Decrease \emph{or} Increase}
\label{subsec:kroening}

The first acceleration technique discussed in this section exploits the following observation:
If $\phi(\vec{a}(\vec{x}))$ implies $\phi(\vec{x})$ and if $\phi(\vec{a}^{n-1}(\vec{x}))$ holds, then the loop condition $\phi$ of \ref{loop} is satisfied throughout (at least) $n$ loop iterations.
So in other words, it requires that the indicator function (or characteristic function) $\charfun{\phi}: \ZZ^d \to \{0,1\}$ of $\phi$ with $\charfun{\phi}(\vec{x}) = 1 \iff \phi(\vec{x})$ is monotonically decreasing w.r.t.\ $\vec{a}$, i.e., $\charfun{\phi}(\vec{x}) \geq \charfun{\phi}(\vec{a}(\vec{x}))$.
\begin{theorem}[Acceleration via Monotonic Decrease \cite{underapprox15}]
  \label{thm:one-way}
  If
  \[
    \phi(\vec{a}(\vec{x})) \implies \phi(\vec{x}),
  \]
  then the following acceleration technique is exact:
  \[
    \ref{loop} \mapsto \vec{x}' = \vec{a}^n(\vec{x}) \land \phi(\vec{a}^{n-1}(\vec{x}))
  \]
\end{theorem}
\noindent
We will prove the more general \Cref{thm:conditional-one-way} in \Cref{sec:conditional}.

So for example, \Cref{thm:one-way} accelerates \ref{loop:exp} to \ref{psi:exp}.
However, the requirement $\phi(\vec{a}(\vec{x})) \implies \phi(\vec{x})$ is often violated in practice.
To see this, consider the loop
\begin{equation}
  \label{eq:conditional-ex}\tag{\ensuremath{\TT_{non\text{-}dec}}}
  \resizebox{0.355\textwidth}{!}{$\mWhile{x_1 > 0 \land x_2 > 0}{\mat{x_1\\x_2} \assign \mat{x_1 - 1\\x_2 + 1}}$}.
\end{equation}
It cannot be accelerated with \Cref{thm:one-way} as
\[
  x_1-1 > 0 \land x_2+1 > 0 \centernot\implies x_1 > 0 \land x_2 > 0.
\]

A dual acceleration technique to \Cref{thm:one-way} is obtained by ``reversing'' the implication in the prerequisites of the theorem, i.e., by requiring
\[
  \phi(\vec{x}) \implies \phi(\vec{a}(\vec{x})).
\]
So the resulting dual acceleration technique applies iff $\phi$ is a loop invariant of \ref{loop}.\footnote{We call a formula $\delta$ a \emph{loop invariant} of a loop \ref{loop} if $\phi(\vec{x}) \land \delta(\vec{x}) \implies \delta(\vec{a}(\vec{x}))$ is valid.}
Then $\{\vec{x} \in \ZZ^d \mid \phi(\vec{x})\}$ is a \emph{recurrent set} \cite{rupak08} (see also \Cref{subsec:related-nonterm}) of \ref{loop}.
In other words, this acceleration technique applies if $\charfun{\phi}$ is monotonically increasing w.r.t.\ $\vec{a}$.

\begin{theorem}[Acceleration via Monotonic Increase]
  \label{thm:recurrent}
  If
  \[
    \phi(\vec{x}) \implies \phi(\vec{a}(\vec{x})),
  \]
  then the following acceleration technique is exact:
  \[
    \ref{loop} \mapsto \vec{x}' = \vec{a}^n(\vec{x}) \land \phi(\vec{x})
  \]
\end{theorem}
\noindent
We will prove the more general \Cref{thm:conditional-recurrent} in \Cref{sec:conditional}.

\begin{example}
As a minimal example, \Cref{thm:recurrent} accelerates
\begin{equation}
  \label{ex:recurrent}\tag{\ensuremath{\TT_{inc}}}
  \mWhile{x > 0}{x \assign x+1}
\end{equation}
to $x' = x + n \land x > 0$.
\end{example}

\subsection{Acceleration via Decrease \emph{and} Increase}
\label{subsec:three-way}

Both acceleration techniques presented so far have been generalized in \cite{fmcad19}.
\begin{theorem}[Acceleration via Monotonicity \cite{fmcad19}]
  \label{thm:three-way}
  If
  \begin{align*}
    \phi(\vec{x}) \iff& \phi_1(\vec{x}) \land \phi_2(\vec{x}) \land \phi_3(\vec{x}),\\
    \phi_1(\vec{x}) \implies& \phi_1(\vec{a}(\vec{x})),\\
    \phi_1(\vec{x}) \land \phi_2(\vec{a}(\vec{x})) \implies& \phi_2(\vec{x}), \qquad\qquad \text{and} \\
    \phi_1(\vec{x}) \land \phi_2(\vec{x}) \land \phi_3(\vec{x}) \implies& \phi_3(\vec{a}(\vec{x})),
  \end{align*}
  then the following acceleration technique is exact:
  \[
    \ref{loop} \mapsto \vec{x}' = \vec{a}^n(\vec{x}) \land \phi_1(\vec{x}) \land \phi_2(\vec{a}^{n-1}(\vec{x})) \land \phi_3(\vec{x})
  \]
\end{theorem}
\makeproof{thm:three-way}{Immediate consequence of \Cref{cor:calculus-sound} and \Cref{thm:simulate-three-way}, which will be proven in \Cref{sec:integration,sec:conditional}.}

Here, $\phi_1$ and $\phi_3$ are again invariants of the loop.
Thus, as in \Cref{thm:recurrent} it suffices to require that they hold before entering the loop.
On the other hand, $\phi_2$ needs to satisfy a similar condition as in \Cref{thm:one-way}, and thus it suffices to require that $\phi_2$ holds before the last iteration.
In such cases, i.e., if
\[
  \phi_1(\vec{x}) \land \phi_2(\vec{a}(\vec{x})) \implies \phi_2(\vec{x})
\]
is valid, we also say that $\phi_2$ is a \emph{converse invariant} (w.r.t.\ $\phi_1$).
It is easy to see that \Cref{thm:three-way} is equivalent to \Cref{thm:one-way} if $\phi_1 \equiv \phi_3 \equiv \top$ (where $\top$ denotes logical truth) and it is equivalent to \Cref{thm:recurrent} if $\phi_2 \equiv \phi_3 \equiv \top$.

\begin{example}
With \Cref{thm:three-way}, \ref{eq:conditional-ex} can be accelerated to
\begin{equation}
  \label{psi:conditional-ex}
  \tag{\ensuremath{\psi_{non\text{-}dec}}}
  \resizebox{0.355\textwidth}{!}{$\mat{x_1'\\x_2'} = \mat{x_1-n\\x_2+n} \land x_2 > 0 \land x_1-n+1 > 0$}
\end{equation}
by choosing $\phi_1 \Def x_2 > 0$, $\phi_2 \Def x_1 > 0$, and $\phi_3 \Def \top$.
\end{example}

\Cref{thm:three-way} naturally raises the question: Why do we need \emph{two} invariants?
To see this, consider a restriction of \Cref{thm:three-way} where $\phi_3 \Def \top$.
It would fail for a loop like
\begin{equation}
  \label{eq:three-way-ex}
  \tag{\ensuremath{\TT_{2\text{-}invs}}}
  \resizebox{0.37\textwidth}{!}{$\mWhile{x_1 > 0 \land x_2 > 0}{\mat{x_1\\x_2} \assign \mat{x_1+x_2\\x_2-1}}$}
\end{equation}
which can easily be handled by \Cref{thm:three-way} (by choosing $\phi_1 \Def \top$, $\phi_2 \Def x_2 > 0$, and $\phi_3 \Def x_1 > 0$).
The problem is that the converse invariant $x_2 > 0$ is needed to prove invariance of $x_1 > 0$.
Similarly, a restriction of \Cref{thm:three-way} where $\phi_1 \Def \top$ would fail for the following variant of \ref{eq:three-way-ex}:
\[
  \mWhile{x_1 > 0 \land x_2 > 0}{\mat{x_1\\x_2} \assign \mat{x_1-x_2\\x_2+1}}
\]
Here, the problem is that the invariant $x_2 > 0$ is needed to prove converse invariance of $x_1 > 0$.

\subsection{Acceleration via Metering Functions}
\label{sec:metering}

Another approach for loop acceleration uses \emph{metering functions}, a variation of classical \emph{ranking functions} from termination and complexity analysis \cite{ijcar16}.
While ranking functions give rise to \emph{upper} bounds on the runtime of loops, metering functions provide \emph{lower} runtime bounds, i.e., the definition of a metering function $\mf: \ZZ^d \to \QQ$ ensures that for each $\vec{x} \in \ZZ^d$, the loop under consideration can be applied at least $\lceil \mf(\vec{x}) \rceil$ times.

\begin{definition}[Metering Function \cite{ijcar16}]
  \label{def:metering}
  We call a function $\mf: \ZZ^d \to \QQ$ a \emph{metering function} if the following holds:
  \begin{align}
    \phi(\vec{x}) & \implies \mf(\vec{x}) - \mf(\vec{a}(\vec{x})) \leq 1 \text{ and} \notag \\
    \neg\phi(\vec{x}) & \implies \mf(\vec{x}) \leq 0 \label{eq:mf-bounded} \tag{mf-bounded}
  \end{align}
\end{definition}

We can use metering functions to accelerate loops as follows:

\begin{theorem}[Acceleration via Metering Functions \cite{ijcar16,journal}]
  \label{thm:meter}
  Let $\mf$ be a metering function for \ref{loop}.
  Then the following acceleration technique is sound:
  \[
    \ref{loop} \mapsto \vec{x}' = \vec{a}^n(\vec{x}) \land n < \mf(\vec{x}) + 1
  \]
\end{theorem}
\noindent
We will prove the more general \Cref{thm:conditional-metering} in \Cref{sec:conditional}.
In contrast to \cite[Thm.\ 3.8]{journal} and \cite[Thm.\ 7]{ijcar16}, the acceleration technique from \Cref{thm:meter} does not conjoin the loop condition $\phi$ to the result, which turned out to be superfluous.
The reason is that $0 < n < \mf(\vec{x}) + 1$ implies $\phi$ due to \eqref{eq:mf-bounded}.

\begin{example}
Using the metering function $(x_1,x_2) \mapsto x_1$, \Cref{thm:meter} accelerates \ref{loop:exp} to
\[
  \left(\mat{x_1'\\x_2'} = \mat{x_1-n\\2^n \cdot x_2} \land n < x_1 + 1\right) \equiv \ref{psi:exp}.
\]
\end{example}

However, synthesizing non-trivial (i.e., non-constant) metering functions is challenging.
Moreover, unless the number of iterations of \ref{loop} equals $\lceil \mf(\vec{x}) \rceil$ for all $\vec{x} \in \ZZ^d$, \emph{acceleration via metering functions} is not exact.

\emph{Linear} metering functions can be synthesized via Farkas' Lemma and SMT solving \cite{ijcar16}.
However, many loops have only trivial linear metering functions.
To see this, reconsider \ref{eq:conditional-ex}.
Here, $(x_1,x_2) \mapsto x_1$ is not a metering function as \ref{eq:conditional-ex} cannot be iterated at least $x_1$ times if $x_2 \leq 0$.
Thus, \cite{journal} proposes a refinement of \cite{ijcar16} based on metering functions of the form $\vec{x} \mapsto \charfun{\xi}(\vec{x}) \cdot f(\vec{x})$ where $\xi \in \Prop{\AAA(\vec{x})}$ and $f$ is linear.
With this improvement, the metering function
\[
  (x_1,x_2) \mapsto \charfun{x_2 > 0}(x_2) \cdot x_1
\]
can be used to accelerate \ref{eq:conditional-ex} to
\[
  \mat{x_1'\\x_2'} = \mat{x_1-n\\x_2+n} \land x_2 > 0 \land n < x_1 + 1.
\]

\section{A Calculus for Modular Loop Acceleration}
\label{sec:integration}

All acceleration techniques presented so far are monolithic:
Either they accelerate a loop successfully or they fail completely.
In other words, we cannot \emph{combine} several techniques to accelerate a single loop.
To this end, we now present a calculus that repeatedly applies acceleration techniques to simplify an \emph{acceleration problem} resulting from a loop \ref{loop} until it is \emph{solved} and hence gives rise to a suitable $\psi \in \Prop{\AAA(\vec{y})}$ which approximates or is equivalent to \ref{loop}.
\begin{definition}[Acceleration Problem]
  \label{def:acceleration_problem}
  A tuple
  \[
    \prob{\psi}{\check{\phi}}{\hat{\phi}}{\vec{a}}
  \]
  where $\psi \in \Prop{\AAA(\vec{y})}$, $\check{\phi}, \hat{\phi} \in \Prop{\AAA(\vec{x})}$, and $\vec{a}: \ZZ^d \to \ZZ^d$ is an \emph{acceleration problem}.
  It is \emph{consistent} if $\psi$ approximates $\langle \check{\phi}, \vec{a} \rangle$, \emph{exact} if $\psi$ is equivalent to $\langle \check{\phi}, \vec{a} \rangle$, and \emph{solved} if it is consistent and $\hat{\phi} \equiv \top$.
  The \emph{canonical acceleration problem} of a loop \ref{loop} is
  \[
    \prob{\vec{x}' = \vec{a}^n(\vec{x})}{\top}{\phi(\vec{x})}{\vec{a}}.
  \]
\end{definition}
\begin{example}
  \label{ex:canonical}
  The canonical acceleration problem of \ref{eq:conditional-ex} is
  \[
    \prob{\mat{x_1'\\x_2'} = \mat{x_1-n\\x_2+n}}{\top}{x_1>0 \land x_2>0}{\mat{x_1-1\\x_2+1}}.
  \]
\end{example}

The first component $\psi$ of an acceleration problem $\prob{\psi}{\check{\phi}}{\hat{\phi}}{\vec{a}}$ is the partial result that has been computed so far.
The second component $\check{\phi}$ corresponds to the part of the loop condition that has already been processed successfully.
As our calculus preserves consistency, $\psi$ always approximates $\langle \check{\phi}, \vec{a} \rangle$.
The third component is the part of the loop condition that remains to be processed, i.e., the loop $\langle \hat{\phi},\vec{a} \rangle$ still needs to be accelerated.
The goal of our calculus is to transform a canonical into a solved acceleration problem.

More specifically, whenever we have simplified a canonical acceleration problem
\[
  \prob{\vec{x}' = \vec{a}^n(\vec{x})}{\top}{\phi(\vec{x})}{\vec{a}}
\]
to
\[
  \prob{\psi_1(\vec{y})}{\check{\phi}(\vec{x})}{\hat{\phi}(\vec{x})}{\vec{a}},
\]
then $\phi \equiv \check{\phi} \land \hat{\phi}$ and
\[
  \psi_1 \quad \text{implies} \quad \vec{x} \longrightarrow^n_{\langle \check{\phi}, \vec{a} \rangle} \vec{x}'.
\]
Then it suffices to find some $\psi_2 \in \Prop{\AAA(\vec{y})}$ such that
\begin{equation}
  \label{eq:post}
  \vec{x} \longrightarrow^n_{\langle \check{\phi}, \vec{a} \rangle} \vec{x}' \land \psi_2 \quad \text{implies} \quad \vec{x} \longrightarrow^n_{\langle \hat{\phi}, \vec{a} \rangle} \vec{x}'.
\end{equation}
The reason is that we have ${\longrightarrow_{\langle \check{\phi}, \vec{a} \rangle}} \cap {\longrightarrow_{\langle \hat{\phi}, \vec{a} \rangle}} = {\longrightarrow_{\langle \check \phi \land \hat \phi, \vec{a} \rangle}} = {\longrightarrow_{\langle \phi, \vec{a} \rangle}}$ and thus
\[
  \psi_1 \land \psi_2 \quad \text{implies} \quad \vec{x} \longrightarrow^n_{\langle \phi, \vec{a} \rangle} \vec{x}',
\]
i.e., $\psi_1 \land \psi_2$ approximates \ref{loop}.

Note that the acceleration techniques presented so far would map $\langle \hat{\phi}, \vec{a} \rangle$ to some $\psi_2 \in \Prop{\AAA(\vec{y})}$ such that
\begin{equation}
  \label{eq:current}
  \psi_2 \quad \text{implies} \quad \vec{x} \longrightarrow^n_{\langle \hat{\phi}, \vec{a} \rangle} \vec{x}',
\end{equation}
which does not use the information that we have already accelerated $\langle \check{\phi}, \vec{a} \rangle$.
In \Cref{sec:conditional}, we will adapt all acceleration techniques from \Cref{sec:monotonic} to search for some $\psi_2 \in \Prop{\AAA(\vec{y})}$ that satisfies \eqref{eq:post} instead of \eqref{eq:current}, i.e., we will turn them into \emph{conditional acceleration techniques}.
\begin{definition}[Conditional Acceleration]
  \label{def:cond-accel}
  We call a partial function
  \[
    \accelerate: \MLoop \times \Prop{\AAA(\vec{x})} \rightharpoonup \Prop{\AAA(\vec{y})}
  \]
  a \emph{conditional acceleration technique}.
  It is \emph{sound} if
  \[
    \vec{x} \longrightarrow^n_{\langle \check{\phi}, \vec{a} \rangle} \vec{x}' \land \accelerate(\langle \chi, \vec{a} \rangle,\check{\phi}) \quad \text{implies} \quad \vec{x} \longrightarrow^n_{\langle \chi, \vec{a} \rangle} \vec{x}'
  \]
  for all $(\langle \chi, \vec{a} \rangle,\check{\phi}) \in \dom(\accelerate)$, $\vec{x},\vec{x}' \in \ZZ^d$, and $n > 0$.
  It is \emph{exact} if additionally
  \[
    \vec{x} \longrightarrow^n_{\langle \chi \land \check{\phi}, \vec{a} \rangle} \vec{x}' \quad \text{implies} \quad \accelerate(\langle \chi, \vec{a} \rangle,\check{\phi})
  \]
  for all $(\langle \chi, \vec{a} \rangle,\check{\phi}) \in \dom(\accelerate)$, $\vec{x},\vec{x}' \in \ZZ^d$, and $n > 0$.
\end{definition}

Note that every acceleration technique gives rise to a conditional acceleration technique in a straightforward way (by disregarding the second argument $\check{\phi}$ of $\accelerate$ in \Cref{def:cond-accel}).
Soundness and exactness can be lifted directly to the conditional setting.
\begin{lemma}[Acceleration as Conditional Acceleration]
  \label{thm:accel-as-cond-accel}
  Let $\accelerate_0$ be an acceleration technique following \Cref{def:accel} such that $\accelerate_0(\langle \chi, \vec{a} \rangle) \implies \vec{x}' = \vec{a}^n(\vec{x})$ is valid whenever $\langle \chi, \vec{a} \rangle \in \dom(\accelerate_0)$.
  Then for the conditional acceleration technique $\accelerate$ given by $\accelerate(\TT,\check\phi) \Def \accelerate_0(\TT)$, the following holds:
  \begin{enumerate}
    \item $\accelerate$ is sound if and only if $\accelerate_0$ is sound
    \item $\accelerate$ is exact if and only if $\accelerate_0$ is exact
  \end{enumerate}
\end{lemma}
\makeproof{thm:accel-as-cond-accel}{
  For the ``if'' direction of 1., we need to show that
  \[
    \vec{x} \longrightarrow^n_{\langle \check{\phi}, \vec{a} \rangle} \vec{x}' \land \accelerate(\langle \chi, \vec{a} \rangle,\check{\phi}) \quad \text{implies} \quad \vec{x} \longrightarrow^n_{\langle \chi, \vec{a} \rangle} \vec{x}'
  \]
  if $\accelerate_0$ is a sound acceleration technique. Thus:
  \begin{align*}
    &\vec{x} \longrightarrow^n_{\langle \check{\phi}, \vec{a} \rangle} \vec{x}' \land \accelerate(\langle \chi, \vec{a} \rangle,\check{\phi}) \\
    \implies & \accelerate(\langle \chi, \vec{a} \rangle,\check{\phi}) \\
    \iff & \accelerate_0(\langle \chi, \vec{a} \rangle) \tag{by definition of $\accelerate$}\\
    \implies & \vec{x} \longrightarrow^n_{\langle \chi, \vec{a} \rangle} \vec{x}' \tag{by soundness of $\accelerate_0$} \\
  \end{align*}

  For the ``only if'' direction of 1., we need to show that
  \[
    \accelerate_0(\langle \phi,\vec{a} \rangle) \quad \text{implies} \quad \vec{x} \longrightarrow^n_{\langle \phi, \vec{a} \rangle} \vec{x}'
  \]
  if $\accelerate$ is a sound conditional acceleration technique.
  Thus:
  \begin{align*}
    & \accelerate_0(\langle \phi,\vec{a} \rangle) \\
    \iff & \accelerate_0(\langle \phi,\vec{a} \rangle) \land \vec{x}' = \vec{a}^{n}(\vec{x}) \tag{by the prerequisites of the lemma}\\
    \iff & \vec{x} \longrightarrow^n_{\langle \top, \vec{a} \rangle} \vec{x}' \land \accelerate_0(\langle \phi,\vec{a} \rangle) \\
    \iff & \vec{x} \longrightarrow^n_{\langle \top, \vec{a} \rangle} \vec{x}' \land \accelerate(\langle \phi,\vec{a} \rangle, \top) \tag{by definition of $\accelerate$}\\
    \implies & \vec{x} \longrightarrow^n_{\langle \phi, \vec{a} \rangle} \vec{x}' \tag{by soundness of $\accelerate$} \\
  \end{align*}

  For the ``if'' direction of 2., soundness of $\accelerate$ follows from 1. We still need to show that
  \[
    \vec{x} \longrightarrow^n_{\langle \chi \land \check{\phi}, \vec{a} \rangle} \vec{x}' \quad \text{implies} \quad \accelerate(\langle \chi, \vec{a} \rangle,\check{\phi})
  \]
  if $\accelerate_0$ is an exact acceleration technique. Thus:
  \begin{align*}
    &\vec{x} \longrightarrow^n_{\langle \chi \land \check{\phi}, \vec{a} \rangle} \vec{x}'\\
  \implies &\vec{x} \longrightarrow^n_{\langle \chi, \vec{a} \rangle} \vec{x}'\\
  \iff & \accelerate_0(\langle \chi, \vec{a} \rangle) \tag{by exactness of $\accelerate_0$}\\
  \iff & \accelerate(\langle \chi, \vec{a} \rangle,\check{\phi}) \tag{by definition of $\accelerate$}
  \end{align*}
  
  For the ``only if'' direction of 2., soundness of $\accelerate_0$ follows from 1. We still need to show that
  \[
    \vec{x} \longrightarrow^n_{\langle \phi, \vec{a} \rangle} \vec{x}' \quad \text{implies} \quad \accelerate_0(\langle \phi,\vec{a} \rangle)
  \]
  if $\accelerate$ is an exact conditional acceleration technique. Thus:
  \begin{align*}
    &\vec{x} \longrightarrow^n_{\langle \phi, \vec{a} \rangle} \vec{x}'\\
    \implies & \accelerate(\langle \phi, \vec{a} \rangle, \top) \tag{by exactness of $\accelerate$} \\
    \iff & \accelerate_0(\langle \phi, \vec{a} \rangle) \tag{by definition of $\accelerate$} \\
  \end{align*}
}

We are now ready to present our \emph{acceleration calculus}, which combines loop acceleration techniques in a modular way.
In the following, w.l.o.g.\ we assume that formulas are in CNF, and we identify the formula $\bigwedge_{i=1}^k C_i$ with the set of clauses $\{C_i \mid 1 \leq i \leq k\}$.
\begin{definition}[Acceleration Calculus]
  \label{def:calculus}
  The relation ${\leadsto}$ on acceleration problems is defined by the rule
  \[
    \infer{
      \prob{\psi_1}{\check{\phi}}{\hat{\phi}}{\vec{a}} \leadsto_{(e)} \prob{\psi_1 \cup \psi_2}{\check{\phi} \cup \chi}{\hat{\phi} \setminus \chi}{\vec{a}}
    }{
      \emptyset \neq \chi \subseteq \hat{\phi} & \accelerate(\langle \chi, \vec{a} \rangle, \check{\phi}) = \psi_2
    }
  \]
  where $\accelerate$ is a sound conditional acceleration technique.
  A ${\leadsto}$-step is \emph{exact} (written ${\leadsto_e}$) if $\accelerate$ is exact.
\end{definition}

So our calculus allows us to pick a subset $\chi$ (of clauses) from the yet unprocessed condition $\hat{\phi}$ and ``move'' it to $\check\phi$, which has already been processed successfully.
To this end, $\langle \chi, \vec{a} \rangle$ needs to be accelerated by a conditional acceleration technique, i.e., when accelerating $\langle \chi, \vec{a} \rangle$ we may assume $\vec{x} \longrightarrow_{\langle \check\phi, \vec{a} \rangle}^n \vec{x}'$.

With \Cref{thm:accel-as-cond-accel}, our calculus allows for combining arbitrary existing acceleration techniques without adapting them.
However, many acceleration techniques can easily be turned into more sophisticated conditional acceleration techniques (see \Cref{sec:conditional}), which increases the power of our approach.

\begin{figure*}
    \begin{align*}
      & \prob{\psi^{init}_{non\text{-}dec} \Def \mat{x_1'\\x_2'} = \mat{x_1-n\\x_2+n}}{\top}{x_1>0 \land x_2>0}{\mat{x_1-1\\x_2+1}} \\
      \leadsto_e & \prob{\psi^{init}_{non\text{-}dec} \land x_1-n+1 > 0}{x_1>0}{x_2>0}{\mat{x_1-1\\x_2+1}} \tag{\Cref{thm:one-way}} \\
      \leadsto_e & \prob{\psi^{init}_{non\text{-}dec} \land x_1-n+1 > 0 \land x_2 > 0}{x_1>0 \land x_2>0}{\top}{\mat{x_1-1\\x_2+1}} \tag{\Cref{thm:recurrent}}\\
      {} = {} & \prob{\ref{psi:conditional-ex}}{x_1 > 0 \land x_2 > 0}{\top}{\mat{x_1-1\\x_2+1}}
    \end{align*}
  \caption{${\leadsto}$-derivation for \ref{eq:conditional-ex}}
  \label{fig:deriv-calculus}
\end{figure*}

\begin{example}
  \label{ex:calculus}
  We continue \Cref{ex:canonical}, where we fix $\chi \Def x_1>0$ for the first acceleration step.
  Thus, we first need to accelerate the loop $\left\langle x_1>0, \mat{x_1-1\\x_2+1} \right\rangle$ to enable a first $\leadsto$-step, and we need to accelerate $\left\langle x_2>0, \mat{x_1-1\\x_2+1} \right\rangle$ afterwards.
  The resulting derivation is shown in \Cref{fig:deriv-calculus}.
  Thus, we successfully constructed the formula \ref{psi:conditional-ex}, which is equivalent to \ref{eq:conditional-ex}.
  Note that here neither of the two steps benefit from the component $\hat{\phi}$ of the acceleration problems.
  We will introduce more powerful conditional acceleration techniques that benefit from $\hat{\phi}$ in \Cref{sec:conditional}.
\end{example}

The crucial property of our calculus is the following.
\begin{lemma}
  \label{thm:calculus-sound}
  The relation $\leadsto$ preserves consistency, and the relation $\leadsto_e$ preserves exactness.
\end{lemma}
\makeproof{thm:calculus-sound}{
  For the first part of the lemma, assume
  \[
    \prob{\psi_1}{\check{\phi}}{\hat{\phi}}{\vec{a}} \leadsto \prob{\psi_1 \cup \psi_2}{\check{\phi} \cup \chi}{\hat{\phi} \setminus \chi}{\vec{a}}
  \]
  where $\prob{\psi_1}{\check{\phi}}{\hat{\phi}}{\vec{a}}$ is consistent and
  \[
    \accelerate(\langle \chi, \vec{a} \rangle, \check{\phi}) = \psi_2.
  \]
  We get
  \begin{align*}
    &\psi_1 \land \psi_2\\
    \implies& \vec{x} \longrightarrow^n_{\langle \check{\phi}, \vec{a} \rangle} \vec{x}' \land \psi_2 \\
    \implies& \vec{x} \longrightarrow^n_{\langle \check{\phi}, \vec{a} \rangle} \vec{x}' \land \vec{x} \longrightarrow^n_{\langle \chi, \vec{a} \rangle} \vec{x}' \\
    \iff& \vec{x} \longrightarrow^n_{\langle \check{\phi} \land \chi, \vec{a} \rangle} \vec{x}'
  \end{align*}
  The first step holds since $\prob{\psi_1}{\check{\phi}}{\hat{\phi}}{\vec{a}}$ is consistent and the second step holds since $\accelerate$ is sound.
  This proves consistency of
  \begin{align*}
    &\prob{\psi_1 \land \psi_2}{\check{\phi} \land \chi}{\hat{\phi} \setminus \chi}{\vec{a}}\\
    = & \prob{\psi_1 \cup \psi_2}{\check{\phi} \cup \chi}{\hat{\phi} \setminus \chi}{\vec{a}}.
  \end{align*}

  For the second part of the lemma, assume
  \[
    \prob{\psi_1}{\check{\phi}}{\hat{\phi}}{\vec{a}} \leadsto_e \prob{\psi_1 \cup \psi_2}{\check{\phi} \cup \chi}{\hat{\phi} \setminus \chi}{\vec{a}}
  \]
  where $\prob{\psi_1}{\check{\phi}}{\hat{\phi}}{\vec{a}}$ is exact and $\accelerate(\langle \chi, \vec{a} \rangle, \check{\phi}) = \psi_2$.
  We get
  \begin{align*}
    &\vec{x} \longrightarrow^n_{\langle \check{\phi} \land \chi, \vec{a} \rangle} \vec{x}'\\
    \iff & \vec{x} \longrightarrow^n_{\langle \check{\phi} \land \chi, \vec{a} \rangle} \vec{x}' \land \psi_2 \tag{by exactness of $\accelerate$}\\
    \iff & \vec{x} \longrightarrow^n_{\langle \check{\phi}, \vec{a} \rangle} \vec{x}' \land \psi_2\\
    \iff & \psi_1 \land \psi_2 \tag{by exactness of $\prob{\psi_1}{\check{\phi}}{\hat{\phi}}{\vec{a}}$}\\
  \end{align*}
  which proves exactness of
  \[
     \prob{\psi_1 \cup \psi_2}{\check{\phi} \cup \chi}{\hat{\phi} \setminus \chi}{\vec{a}}.
  \]
}

Then the correctness of our calculus follows immediately.
The reason is that
\[
  \prob{\vec{x}' = \vec{a}^n(\vec{x})}{\top}{\phi(\vec{x})}{\vec{a}} \leadsto_{(e)}^* \prob{\psi(\vec{y})}{\check{\phi}(\vec{x})}{\top}{\vec{a}}
\]
implies $\phi \equiv \check{\phi}$.

\begin{theorem}[Correctness of ${\leadsto}$]
  \label{cor:calculus-sound}
  If
  \[
    \prob{\vec{x}' = \vec{a}^n(\vec{x})}{\top}{\phi(\vec{x})}{\vec{a}} \leadsto^* \prob{\psi(\vec{y})}{\check{\phi}(\vec{x})}{\top}{\vec{a}},
  \]
  then $\psi$ approximates \ref{loop}.
  If
  \[
    \prob{\vec{x}' = \vec{a}^n(\vec{x})}{\top}{\phi(\vec{x})}{\vec{a}} \leadsto_e^* \prob{\psi(\vec{y})}{\check{\phi}(\vec{x})}{\top}{\vec{a}},
  \]
  then $\psi$ is equivalent to \ref{loop}.
\end{theorem}

Termination of our calculus is trivial, as the size of the third component $\hat{\phi}$ of the acceleration problem is decreasing.

\begin{theorem}[Well-Foundedness of ${\leadsto}$]
  \label{thm:term}
  The relation ${\leadsto}$ is well-founded.
\end{theorem}

\section{Conditional Acceleration Techniques}
\label{sec:conditional}

We now show how to turn the acceleration techniques from \Cref{sec:monotonic} into conditional acceleration techniques, starting with \emph{acceleration via monotonic decrease}.
\begin{theorem}[Conditional Acceleration via Monotonic Decrease]
  \label{thm:conditional-one-way}
  If
  \begin{equation}
    \label{eq:conditional-one-way}
    \submission{\notag}
    \check{\phi}(\vec{x}) \land \chi(\vec{a}(\vec{x})) \implies \chi(\vec{x}),
  \end{equation}
  then the following conditional acceleration technique is exact:
  \[
    (\langle \chi, \vec{a} \rangle, \check{\phi}) \mapsto \vec{x}' = \vec{a}^n(\vec{x}) \land \chi(\vec{a}^{n-1}(\vec{x}))
  \]
\end{theorem}
\makeproof{thm:conditional-one-way}{
  For soundness, we need to prove
  \begin{multline}
    \label{thm:conditional-one-way-ih}
    \vec{x} \longrightarrow^m_{\langle \check{\phi}, \vec{a} \rangle} \vec{a}^{m}(\vec{x}) \land \chi(\vec{a}^{m-1}(\vec{x})) \\ \implies \vec{x} \longrightarrow^m_{\langle \chi, \vec{a} \rangle} \vec{a}^{m}(\vec{x})
  \end{multline}
  for all $m > 0$.
  We use induction on $m$.
  If $m=1$, then
  \begin{align*}
    &\vec{x} \longrightarrow^m_{\langle \check{\phi}, \vec{a} \rangle} \vec{a}^{m}(\vec{x}) \land \chi(\vec{a}^{m-1}(\vec{x}))\\
    \implies & \chi(\vec{x}) \tag{as $m=1$}\\
    \iff & \vec{x} \longrightarrow_{\langle \chi, \vec{a} \rangle} \vec{a}(\vec{x})\\
    \iff & \vec{x} \longrightarrow^m_{\langle \chi, \vec{a} \rangle} \vec{a}^m(\vec{x}). \tag{as $m=1$}
  \end{align*}
  In the induction step, we have
  \begin{align*}
    &\vec{x} \longrightarrow^{m+1}_{\langle \check{\phi}, \vec{a} \rangle} \vec{a}^{m+1}(\vec{x}) \land \chi(\vec{a}^m(\vec{x})) \\
    \implies & \vec{x} \longrightarrow^{m}_{\langle \check{\phi}, \vec{a} \rangle} \vec{a}^{m}(\vec{x}) \land \chi(\vec{a}^m(\vec{x}))\\
    \iff & \vec{x} \longrightarrow^{m}_{\langle \check{\phi}, \vec{a} \rangle} \vec{a}^{m}(\vec{x}) \land \check{\phi}(\vec{a}^{m-1}(\vec{x})) \land \chi(\vec{a}^m(\vec{x})) \tag{as $m > 0$}\\
    \implies & \vec{x} \longrightarrow^{m}_{\langle \check{\phi}, \vec{a} \rangle} \vec{a}^{m}(\vec{x}) \land \chi(\vec{a}^{m-1}(\vec{x})) \land \chi(\vec{a}^m(\vec{x})) \tag{due to \eqref{eq:conditional-one-way}} \\
    \implies & \vec{x} \longrightarrow^m_{\langle \chi, \vec{a} \rangle} \vec{a}^m(\vec{x}) \land \chi(\vec{a}^m(\vec{x})) \tag{by the induction hypothesis \eqref{thm:conditional-one-way-ih}}\\
    \iff & \vec{x} \longrightarrow^{m+1}_{\langle \chi, \vec{a} \rangle} \vec{a}^{m+1}(\vec{x}).
  \end{align*}

  For exactness, we need to prove
  \[
    \vec{x} \longrightarrow^m_{\langle \chi \land \check{\phi}, \vec{a} \rangle} \vec{a}^m(\vec{x}) \implies \chi(\vec{a}^{m-1}(\vec{x}))
  \]
  for all $m>0$, which is trivial.
}

So we just add $\check{\phi}$ to the premise of the implication that needs to be checked to apply \emph{acceleration via monotonic decrease}.
\Cref{thm:recurrent} can be adapted analogously.
\begin{theorem}[Conditional Acceleration via Monotonic Increase]
  \label{thm:conditional-recurrent}
  If
  \begin{equation}
    \label{eq:conditional-recurrent}
    \submission{\notag}
    \check{\phi}(\vec{x}) \land \chi(\vec{x}) \implies \chi(\vec{a}(\vec{x})),
  \end{equation}
  then the following conditional acceleration technique is exact:
  \[
    (\langle \chi, \vec{a} \rangle, \check{\phi}) \mapsto \vec{x}' = \vec{a}^n(\vec{x}) \land \chi(\vec{x})
  \]
\end{theorem}
\makeproof{thm:conditional-recurrent}{
  For soundness, we need to prove
  \begin{equation}
    \label{eq:conditional-recurrent-ih}
    \vec{x} \longrightarrow^m_{\langle \check\phi, \vec{a} \rangle} \vec{a}^m(\vec{x}) \land \chi(\vec{x}) \implies \vec{x} \longrightarrow^m_{\langle \chi, \vec{a} \rangle} \vec{a}^m(\vec{x})
  \end{equation}
  for all $m>0$.
  We use induction on $m$. If $m=1$, then
  \begin{align*}
    & \vec{x} \longrightarrow^m_{\langle \check\phi, \vec{a} \rangle} \vec{a}^m(\vec{x}) \land \chi(\vec{x})\\
    \implies & \vec{x} \longrightarrow_{\langle \chi, \vec{a} \rangle} \vec{a}(\vec{x}) \\
    \iff & \vec{x} \longrightarrow^m_{\langle \chi, \vec{a} \rangle} \vec{a}^m(\vec{x}).\tag{as $m=1$}
  \end{align*}
  In the induction step, we have
  \begin{align*}
    & \vec{x} \longrightarrow^{m+1}_{\langle \check\phi, \vec{a} \rangle} \vec{a}^{m+1}(\vec{x}) \land \chi(\vec{x})\\
    \implies & \vec{x} \longrightarrow^{m}_{\langle \check\phi, \vec{a} \rangle} \vec{a}^{m}(\vec{x}) \land \chi(\vec{x})\\
    \implies & \vec{x} \longrightarrow^{m}_{\langle \check\phi, \vec{a} \rangle} \vec{a}^{m}(\vec{x}) \land \vec{x} \longrightarrow^m_{\langle \chi, \vec{a} \rangle} \vec{a}^m(\vec{x}) \tag{by the induction hypothesis \eqref{eq:conditional-recurrent-ih}} \\
    \implies & \vec{x} \longrightarrow^m_{\langle \chi, \vec{a} \rangle} \vec{a}^m(\vec{x}) \land \check\phi(\vec{a}^{m-1}(\vec{x})) \land \chi(\vec{a}^{m-1}(\vec{x})) \tag{as $m > 0$}\\
    \implies & \vec{x} \longrightarrow^m_{\langle \chi, \vec{a} \rangle} \vec{a}^m(\vec{x}) \land \chi(\vec{a}^{m}(\vec{x})) \tag{due to \eqref{eq:conditional-recurrent}} \\
    \iff & \vec{x} \longrightarrow^{m+1}_{\langle \chi, \vec{a} \rangle} \vec{a}^{m+1}(\vec{x}).
  \end{align*}

  For exactness, we need to prove
  \[
    \vec{x} \longrightarrow^m_{\langle \chi \land \check\phi, \vec{a} \rangle} \vec{a}^m(\vec{x}) \implies \chi(\vec{x}),
  \]
  for all $m>0$, which is trivial.
}

\begin{figure*}
  \begin{align*}
    & \prob{\vec{x}' = \vec{a}_{2\text{-}invs}^n(\vec{x})}{\top}{x_1>0 \land x_2>0}{\vec{a}_{2\text{-}invs}} \\
    \leadsto_e & \prob{\vec{x}' = \vec{a}_{2\text{-}invs}^n(\vec{x}) \land x_2 - n + 1 > 0}{x_2>0}{x_1>0}{\vec{a}_{2\text{-}invs}} \tag{\Cref{thm:conditional-one-way}} \\
    \leadsto_e & \prob{\vec{x}' = \vec{a}_{2\text{-}invs}^n(\vec{x}) \land x_2 - n + 1 > 0 \land x_1 > 0}{x_2>0 \land x_1>0}{\top}{\vec{a}_{2\text{-}invs}} \tag{\Cref{thm:conditional-recurrent}}
  \end{align*}
  \caption{${\leadsto}$-derivation for \ref{eq:three-way-ex}}
  \label{fig:deriv-three-way}
\end{figure*}

\begin{example}
  For the canonical acceleration problem of \ref{eq:three-way-ex}, we obtain the derivation shown in \Cref{fig:deriv-three-way}, where $\vec{a}_{2\text{-}invs} \Def \mat{x_1+x_2\\x_2-1}$.
  While we could also use \Cref{thm:one-way} for the first step, \Cref{thm:recurrent} is inapplicable in the second step.
  The reason is that we need the converse invariant $x_2>0$ to prove invariance of $x_1 > 0$.
\end{example}

It is not a coincidence that \ref{eq:three-way-ex}, which could also be accelerated with \emph{acceleration via monotonicity} (see \Cref{thm:three-way}) directly, can be handled by applying our novel calculus with \Cref{thm:conditional-one-way,thm:conditional-recurrent}.
\begin{remark}
  \label{thm:simulate-three-way}
  If applying \emph{acceleration via monotonicity} to \ref{loop} yields $\psi$, then
  \[
    \prob{\vec{x}' = \vec{a}^n(\vec{x})}{\top}{\phi(\vec{x})}{\vec{a}}\leadsto_e^{\leq 3} \prob{\psi(\vec{y})}{\phi(\vec{x})}{\top}{\vec{a}}
  \]
  where either \Cref{thm:conditional-one-way} or \Cref{thm:conditional-recurrent} is applied in each ${\leadsto_e}$-step.
\end{remark}
\makeproof{thm:simulate-three-way}{
  As \Cref{thm:three-way} applies, we have
  \[
    \phi(\vec{x}) \equiv \phi_1(\vec{x}) \land \phi_2(\vec{x}) \land \phi_3(\vec{x})
  \]
  where
  \begin{align}
    \phi_1(\vec{x}) &\implies \phi_1(\vec{a}(\vec{x})) & \land \label{eq:simulate-three-way1} \\
    \phi_1(\vec{x}) \land \phi_2(\vec{a}(\vec{x})) &\implies \phi_2(\vec{x}) & \land \label{eq:simulate-three-way2}\\
    \phi_1(\vec{x}) \land \phi_2(\vec{x}) \land \phi_3(\vec{x}) &\implies \phi_3(\vec{a}(\vec{x})). \label{eq:simulate-three-way3}
  \end{align}
  If $\phi_1 \neq \top$, then \Cref{thm:conditional-recurrent} applies to $\langle \phi_1, \vec{a} \rangle$ with $\check{\phi} \Def \top$ due to \eqref{eq:simulate-three-way1}, and we obtain:
  \begin{align*}
    & \prob{\vec{x}' = \vec{a}^n(\vec{x})}{\top}{\phi(\vec{x})}{\vec{a}}\\
    = & \prob{\vec{x}' = \vec{a}^n(\vec{x})}{\top}{\phi_1(\vec{x}) \land \phi_2(\vec{x}) \land \phi_3(\vec{x})}{\vec{a}}\\
    \leadsto_e & \prob{\vec{x}' = \vec{a}^n(\vec{x}) \land \phi_1(\vec{x})}{\phi_1(\vec{x})}{\phi_2(\vec{x}) \land \phi_3(\vec{x})}{\vec{a}} \\
    = & \prob{\psi_1(\vec{y})}{\phi_1(\vec{x})}{\phi_2(\vec{x}) \land \phi_3(\vec{x})}{\vec{a}}
  \end{align*}
  Next, if $\phi_2 \neq \top$, then \Cref{thm:conditional-one-way} applies to $\langle \phi_2, \vec{a} \rangle$ with $\check\phi \Def \phi_1$ due to \eqref{eq:simulate-three-way2} and we obtain:
  \begin{align*}
    &\prob{\psi_1(\vec{y})}{\phi_1(\vec{x})}{\phi_2(\vec{x}) \land \phi_3(\vec{x})}{\vec{a}}\\
    \leadsto_e& \prob{\psi_1(\vec{y}) \land \phi_2(\vec{a}^{n-1}(\vec{x}))}{\phi_1(\vec{x}) \land \phi_2(\vec{x})}{\phi_3(\vec{x})}{\vec{a}} \\
    =& \prob{\psi_2(\vec{y})}{\phi_1(\vec{x}) \land \phi_2(\vec{x})}{\phi_3(\vec{x})}{\vec{a}}
  \end{align*}
  Finally, if $\phi_3 \neq \top$, then \Cref{thm:conditional-recurrent} applies to $\langle \phi_3, \vec{a} \rangle$ with $\check\phi \Def \phi_1 \land \phi_2$ due to \eqref{eq:simulate-three-way3} and we obtain
  \begin{align*}
    &\prob{\psi_2(\vec{y})}{\phi_1(\vec{x}) \land \phi_2(\vec{x})}{\phi_3(\vec{x})}{\vec{a}}\\
    \leadsto_e &\prob{\psi_2(\vec{y}) \land \phi_3(\vec{x})}{\phi(\vec{x})}{\top}{\vec{a}} \\
    =& \prob{\psi(\vec{y})}{\phi(\vec{x})}{\top}{\vec{a}}.
  \end{align*}
}

Thus, there is no need for a conditional variant of \emph{acceleration via monotonicity}.
Note that combining \Cref{thm:conditional-one-way,thm:conditional-recurrent} with our calculus is also useful for loops where \emph{acceleration via monotonicity} is inapplicable.

\begin{figure*}
  \begin{align*}
    & \prob{\psi^{init}_{2\text{-}c\text{-}invs}}{\top}{\phi_{2\text{-}c\text{-}invs}}{\vec{a}_{2\text{-}c\text{-}invs}} \\
    \leadsto_e & \prob{\psi^{init}_{2\text{-}c\text{-}invs} \land x_1^{(n-1)} > 0}{x_1 > 0}{x_2 > 0 \land x_3 > 0}{\vec{a}_{2\text{-}c\text{-}invs}} \tag{\Cref{thm:conditional-one-way}} \\
    \leadsto_e & \prob{\psi^{init}_{2\text{-}c\text{-}invs} \land x_1^{(n-1)} > 0 \land x_2 > 0}{x_1 > 0 \land x_2 > 0}{x_3 > 0}{\vec{a}_{2\text{-}c\text{-}invs}} \tag{\Cref{thm:conditional-recurrent}} \\
    \leadsto_e & \prob{\psi^{init}_{2\text{-}c\text{-}invs} \land x_1^{(n-1)} \! > 0 \land x_2 > 0 \land x_3^{(n-1)} \! > 0}{\phi_{2\text{-}c\text{-}invs}}{\top}{\vec{a}_{2\text{-}c\text{-}invs}} \tag{\Cref{thm:conditional-one-way}}
  \end{align*}
  \caption{${\leadsto}$-derivation for \ref{eq:beyond-three-way}}
  \label{fig:derivation-beyond-three-way}
\end{figure*}

\begin{example}
  Consider the following loop, which can be accelerated by splitting its guard into one invariant and two converse invariants.
  \[
    \label{eq:beyond-three-way}\tag{\ensuremath{\TT_{2\text{-}c\text{-}invs}}}
    \begin{aligned}
      & \mathbf{while}\ x_1 > 0 \land x_2 > 0 \land x_3 > 0\ \mDo\\
      & \qquad \mat{x_1\\x_2\\x_3} \assign \mat{x_1 - 1\\x_2+x_1\\x_3-x_2}
    \end{aligned}
  \]
  Let
  \begin{align*}
    \phi_{2\text{-}c\text{-}invs} &\Def x_1 > 0 \land x_2 > 0 \land x_3 > 0,\\
    \vec{a}_{2\text{-}c\text{-}invs} &\Def \mat{x_1 - 1\\x_2+x_1\\x_3-x_2},\\
    \psi^{init}_{2\text{-}c\text{-}invs} &\Def \vec{x}' = \vec{a}_{2\text{-}c\text{-}invs}^n(\vec{x}),
  \end{align*}
  and let $x_i^{(m)}$ be the $i^{th}$ component of $\vec{a}_{2\text{-}c\text{-}invs}^{m}(\vec{x})$.
  Starting with the canonical acceleration problem of \ref{eq:beyond-three-way}, we obtain the derivation shown in \Cref{fig:derivation-beyond-three-way}.
\end{example}

Finally, we present a variant of \Cref{thm:meter} for conditional acceleration.
The idea is similar to the approach for deducing metering functions of the form $\vec{x} \mapsto \charfun{\check{\phi}}(\vec{x}) \cdot f(\vec{x})$ from \cite{journal} (see \Cref{sec:metering} for details).
But in contrast to \cite{journal}, in our setting the ``conditional'' part $\check{\phi}$ does not need to be an invariant of the loop.
\begin{theorem}[Conditional Acceleration via Metering Functions]
  \label{thm:conditional-metering}
  Let $\mf: \ZZ^d \to \QQ$.
  If
  \begin{align}
    \report{\label{eq:conditional-metering-decrease}}
    \submission{\notag}
    \check{\phi}(\vec{x}) \land \chi(\vec{x}) & \implies \mf(\vec{x}) - \mf(\vec{a}(\vec{x})) \leq 1 \text{ and} \\
    \report{\label{eq:conditional-metering-bound}}
    \submission{\notag}
    \check{\phi}(\vec{x}) \land \neg\chi(\vec{x}) & \implies \mf(\vec{x}) \leq 0,
  \end{align}
  then the following conditional acceleration technique is sound:
  \[
    (\langle \chi, \vec{a} \rangle, \check{\phi}) \mapsto \vec{x}' = \vec{a}^n(\vec{x}) \land n < \mf(\vec{x}) + 1
  \]
\end{theorem}
\makeproof{thm:conditional-metering}{
  We need to prove
  \begin{multline}
    \label{eq:conditional-metering-ih}
    \vec{x} \longrightarrow^m_{\langle \check\phi, \vec{a} \rangle} \vec{a}^{m}(\vec{x}) \land m < \mf(\vec{x}) + 1 \\
    \implies \vec{x} \longrightarrow^m_{\langle \chi, \vec{a} \rangle} \vec{a}^{m}(\vec{x})
  \end{multline}
  for all $m > 0$.
  Note that \eqref{eq:conditional-metering-bound} is equivalent to
  \begin{equation}
    \label{eq:conditional-metering-bound-reversed}
    \mf(\vec{x}) > 0 \implies \neg\check{\phi}(\vec{x}) \lor \chi(\vec{x}).
  \end{equation}
  We use induction on $m$.
  If $m=1$, then
  \begin{align*}
    & \vec{x} \longrightarrow^m_{\langle \check\phi, \vec{a} \rangle} \vec{a}^{m}(\vec{x}) \land m < \mf(\vec{x}) + 1 \\
    \iff & \check\phi(\vec{x}) \land \mf(\vec{x}) > 0 \tag{as $m=1$}\\
    \implies & \chi(\vec{x}) \tag{due to \eqref{eq:conditional-metering-bound-reversed}}\\
    \iff & \vec{x} \longrightarrow_{\langle \chi, \vec{a} \rangle} \vec{a}(\vec{x})\\
    \iff & \vec{x} \longrightarrow^m_{\langle \chi, \vec{a} \rangle} \vec{a}^m(\vec{x}) \tag{as $m=1$}
  \end{align*}
  In the induction step, assume
  \begin{equation}
    \label{eq:conditional-metering-is}
    \vec{x} \longrightarrow^{m+1}_{\langle \check{\phi}, \vec{a} \rangle} \vec{a}^{m+1}(\vec{x}) \land m < \mf(\vec{x}).
  \end{equation}
  Then we have:
  \begin{align*}
    & \eqref{eq:conditional-metering-is} \\
    \implies & \vec{x} \longrightarrow^{m}_{\langle \check{\phi}, \vec{a} \rangle} \vec{a}^{m}(\vec{x}) \land m < \mf(\vec{x}) \\
    \implies & \vec{x} \longrightarrow^m_{\langle \check{\phi}, \vec{a} \rangle} \vec{a}^{m}(\vec{x}) \land m < \mf(\vec{x}) \land {} \\
             & \qquad\qquad\qquad\qquad\quad \vec{x} \longrightarrow^m_{\langle \chi, \vec{a} \rangle} \vec{a}^{m}(\vec{x}) \tag{due to the induction hypothesis \eqref{eq:conditional-metering-ih}} \\
    \iff & m < \mf(\vec{x}) \land \vec{x} \longrightarrow^m_{\langle \chi, \vec{a} \rangle} \vec{a}^{m}(\vec{x})  \land {} \\
    & \qquad \forall i \in [0,m-1].\ \left( \check{\phi}(\vec{a}^i(\vec{x})) \land \chi(\vec{a}^i(\vec{x})) \right)\\
    \implies & m < \mf(\vec{x}) \land \vec{x} \longrightarrow^m_{\langle \chi, \vec{a} \rangle} \vec{a}^{m}(\vec{x}) \land {} \\
    & \qquad \mf(\vec{x}) - \mf(\vec{a}^m(\vec{x})) \leq m \tag{due to \eqref{eq:conditional-metering-decrease}}\\
    \implies & \vec{x} \longrightarrow^m_{\langle \chi, \vec{a} \rangle} \vec{a}^{m}(\vec{x}) \land \mf(\vec{a}^m(\vec{x})) > 0 \tag{as $m < \mf(\vec{x})$}\\
    \implies & \vec{x} \longrightarrow^m_{\langle \chi, \vec{a} \rangle} \vec{a}^{m}(\vec{x}) \land (\neg\check{\phi}(\vec{a}^m(\vec{x})) \lor \chi(\vec{a}^m(\vec{x}))) \tag{due to \eqref{eq:conditional-metering-bound-reversed}}\\
    \iff & \vec{x} \longrightarrow^m_{\langle \chi, \vec{a} \rangle} \vec{a}^{m}(\vec{x}) \land \chi(\vec{a}^m(\vec{x})) \tag{as \eqref{eq:conditional-metering-is} implies $\check{\phi}(\vec{a}^m(\vec{x}))$} \\
    \iff & \vec{x} \longrightarrow^{m+1}_{\langle \chi, \vec{a} \rangle} \vec{a}^{m+1}(\vec{x})
  \end{align*}
}

\section{Acceleration via Eventual Monotonicity}
\label{sec:accel}

The combination of the calculus from \Cref{sec:integration} and the conditional acceleration techniques from \Cref{sec:conditional} still fails to handle certain interesting classes of loops.
Thus, to improve the applicability of our approach we now present two new acceleration techniques based on \emph{eventual} monotonicity.

\subsection{Acceleration via Eventual Decrease}

All (combinations of) techniques presented so far fail for the following example.
\begin{equation}
  \label{eq:phases}\tag{\ensuremath{\TT_{ev\text{-}dec}}}
  \mWhile{x_1 > 0}{\mat{x_1\\x_2} \assign \mat{x_1 + x_2\\x_2 - 1}}
\end{equation}
The reason is that $x_1$ does not behave monotonically, i.e., $x_1 > 0$ is neither an invariant nor a converse invariant.
Essentially, \ref{eq:phases} proceeds in two phases:
In the first (optional) phase, $x_2$ is positive and hence the value of $x_1$ is monotonically increasing.
In the second phase, $x_2$ is non-positive and consequently the value of $x_1$ decreases (weakly) monotonically.
The crucial observation is that once the value of $x_1$ decreases, it can never increase again.
Thus, despite the non-monotonic behavior of $x_1$, it suffices to require that $x_1 > 0$ holds before the first and before the $n^{th}$ loop iteration to ensure that the loop can be iterated at least $n$ times.
\begin{theorem}[Acceleration via Eventual Decrease]
  \label{thm:ev-dec}
  If $\phi(\vec{x}) \equiv \bigwedge_{i=1}^k C_i$ where each clause $C_i$ contains an inequation $\expr_i(\vec{x}) > 0$ such that
  \[
    \expr_i(\vec{x}) \geq \expr_i(\vec{a}(\vec{x})) \implies \expr_i(\vec{a}(\vec{x})) \geq \expr_i(\vec{a}^2(\vec{x})),
  \]
  then the following acceleration technique is sound:
  \[
    \resizebox{0.475\textwidth}{!}{
      $\displaystyle \ref{loop} \mapsto \vec{x}' = \vec{a}^n(\vec{x})\land \bigwedge_{i=1}^k \expr_i(\vec{x}) > 0 \land \expr_i(\vec{a}^{n-1}(\vec{x})) > 0$
    }
  \]
  If $C_i \equiv \expr_i > 0$ for all $i \in [1,k]$, then it is exact.
\end{theorem}
\noindent
We will prove the more general \Cref{thm:ev-dec-cond} later in this section.

\begin{example}
With \Cref{thm:ev-dec}, we can accelerate \ref{eq:phases} to
\begin{align*}
  & \mat{x_1'\\x_2'} = \mat{\tfrac{n-n^2}{2} + x_2 \cdot n + x_1\\x_2 - n} \\
  {} \land {} & x_1 > 0 \\
  {} \land {} & \tfrac{n-1-(n-1)^2}{2} + x_2 \cdot (n-1) + x_1 > 0
\end{align*}
as we have
\begin{align*}
  &(x_1 \geq x_1 + x_2) \equiv (0 \geq x_2) \implies \\
  & \qquad (0 \geq x_2 - 1) \equiv (x_1 + x_2 \geq x_1 + x_2 + x_2 - 1).
\end{align*}
\end{example}

Turning \Cref{thm:ev-dec} into a conditional acceleration technique is straightforward.
\begin{theorem}[Conditional Acceleration via Eventual Decrease]
  \label{thm:ev-dec-cond}
  If we have $\chi(\vec{x}) \equiv \bigwedge_{i=1}^k C_i$ where each clause $C_i$ contains an inequation $\expr_i(\vec{x}) > 0$ such that
  \begin{equation}
    \label{eq:ev-dec-cond-pre}
    \resizebox{0.415\textwidth}{!}{$\check{\phi}(\vec{x}) \land \expr_i(\vec{x}) \geq \expr_i(\vec{a}(\vec{x})) \implies \expr_i(\vec{a}(\vec{x})) \geq \expr_i(\vec{a}^2(\vec{x})),$}
  \end{equation}
  then the following conditional acceleration technique is sound:
  \begin{equation}
    \submission{\notag}
    \label{eq:ev-dec-cond}
    \begin{aligned}
      &(\langle \chi,\vec{a} \rangle, \check\phi) \mapsto \Big( \vec{x}' = \vec{a}^n(\vec{x}) \\
      &\qquad\qquad\quad \land \bigwedge_{i=1}^k \expr_i(\vec{x}) > 0 \land \expr_i(\vec{a}^{n-1}(\vec{x})) > 0 \Big)
    \end{aligned}
  \end{equation}
  If $C_i \equiv \expr_i > 0$ for all $i \in [1,k]$, then it is exact.
\end{theorem}
\makeproof{thm:ev-dec-cond}{
  For soundness, we need to show
  \begin{equation}
    \label{eq:ev-dec-goal-orig}
    \begin{aligned}
      & \Big( \vec{x} \longrightarrow^n_{\langle \check\phi,\vec{a} \rangle} \vec{a}^n(\vec{x}) \\
      & \land \bigwedge_{i=1}^k \expr_i(\vec{x}) > 0 \land \expr_i(\vec{a}^{n-1}(\vec{x})) > 0 \Big) \\
      & \qquad\qquad\qquad \implies \vec{x} \longrightarrow^n_{\langle \chi, \vec{a}\rangle} \vec{a}^n(\vec{x}).
    \end{aligned}
  \end{equation}
  Assume
  \begin{equation}
    \label{eq:ev-dec-assumption-1}
    \vec{x} \longrightarrow^n_{\langle \check\phi,\vec{a} \rangle} \vec{a}^n(\vec{x}) \land \bigwedge_{i=1}^k \expr_i(\vec{x}) > 0 \land \expr_i(\vec{a}^{n-1}(\vec{x})) > 0.
  \end{equation}
  This implies
  \begin{equation}
    \label{eq:ev-dec-assumption2}
    \bigwedge_{i=0}^{n-1} \check\phi(\vec{a}^i(\vec{x})).
  \end{equation}
  In the following, we show
  \begin{equation}
    \label{eq:ev-dec-goal}
    \bigwedge_{i=1}^k\bigwedge_{m=0}^{n-1} \expr_i(\vec{a}^m(\vec{x})) \geq \min(\expr_i(\vec{x}), \expr_i(\vec{a}^{n-1}(\vec{x}))).
  \end{equation}
  Then the claim \eqref{eq:ev-dec-goal-orig} follows, as we have
  \begin{align*}
    &\bigwedge_{i=1}^k\bigwedge_{m=0}^{n-1} \expr_i(\vec{a}^m(\vec{x})) \geq \min(\expr_i(\vec{x}), \expr_i(\vec{a}^{n-1}(\vec{x})))\\
    \implies& \bigwedge_{i=1}^k\bigwedge_{m=0}^{n-1} \expr_i(\vec{a}^m(\vec{x})) > 0 \tag{due to \eqref{eq:ev-dec-assumption-1}}\\
    \implies& \bigwedge_{m=0}^{n-1} \chi(\vec{a}^m(\vec{x})) \tag{by definition of $\expr_i$}\\
    \iff& \vec{x} \longrightarrow^n_{\langle \chi, \vec{a}\rangle} \vec{a}^n(\vec{x}).
  \end{align*}
  Let $i$ be arbitrary but fixed, let $\expr = \expr_i$, and let $j$ be the minimal natural number with
  \begin{equation}
    \label{eq:ev-dec-assumption}
    \expr(\vec{a}^j(\vec{x})) = \max\{\expr(\vec{a}^m(\vec{x})) \mid m \in [0,n-1]\}.
  \end{equation}
  We first prove
  \begin{equation}
    \label{eq:ev-dec-ih}
    \expr(\vec{a}^{m}(\vec{x})) < \expr(\vec{a}^{m+1}(\vec{x}))
  \end{equation}
  for all $m \in [0,j-1]$ by backward induction on $m$.
  If $m=j-1$, then
  \begin{align*}
    & \expr(\vec{a}^{m}(\vec{x})) \\
    = {}& \expr(\vec{a}^{j-1}(\vec{x})) \tag{as $m=j-1$} \\
    < {}& \expr(\vec{a}^{j}(\vec{x})) \tag{due to \eqref{eq:ev-dec-assumption} as $j$ is minimal} \\
    = {}& \expr(\vec{a}^{m+1}(\vec{x})). \tag{as $m=j-1$}
  \end{align*}

  For the induction step, note that \eqref{eq:ev-dec-cond-pre} implies
  \begin{equation}
    \label{eq:ev-dec-cond-pre2}
    \resizebox{0.415\textwidth}{!}{
        $\expr(\vec{a}(\vec{x})) < \expr(\vec{a}^2(\vec{x})) \implies \neg\check\phi(\vec{x}) \lor \expr(\vec{x}) < \expr(\vec{a}(\vec{x})).$
      }
  \end{equation}
  Then we have
  \begin{align*}
    &\expr(\vec{a}^{m+1}(\vec{x})) < \expr(\vec{a}^{m+2}(\vec{x})) \tag{due to the induction hypothesis \eqref{eq:ev-dec-ih}} \\
    \implies & \neg\check\phi(\vec{a}^{m}(\vec{x})) \lor \expr(\vec{a}^{m}(\vec{x})) < \expr(\vec{a}^{m+1}(\vec{x})) \tag{by \eqref{eq:ev-dec-cond-pre2}} \\
    \implies & \expr(\vec{a}^{m}(\vec{x})) < \expr(\vec{a}^{m+1}(\vec{x})). \tag{by \eqref{eq:ev-dec-assumption2}}
  \end{align*}

  Now we prove
  \begin{equation}
    \label{eq:ev-dec-ih2}
    \expr(\vec{a}^m(\vec{x})) \geq \expr(\vec{a}^{m+1}(\vec{x}))
  \end{equation}
  for all $m \in [j,n-1]$ by induction on $m$.
  If $m=j$, then
  \begin{align*}
    & \expr(\vec{a}^{m}(\vec{x})) \\
    = {}& \expr(\vec{a}^{j}(\vec{x})) \tag{as $m=j$} \\
    = {}& \max\{\expr(\vec{a}^m(\vec{x})) \mid m \in [0,n-1]\} \tag{due to \eqref{eq:ev-dec-assumption}}\\
    \geq {}& \expr(\vec{a}^{j+1}(\vec{x})) \\
    = {}& \expr(\vec{a}^{m+1}(\vec{x})).\tag{as $m=j$}
  \end{align*}

  In the induction step, we have
  \begin{align*}
    &\expr(\vec{a}^{m}(\vec{x})) \geq \expr(\vec{a}^{m+1}(\vec{x})) \tag{due to the induction hypothesis \eqref{eq:ev-dec-ih2}} \\
    \implies & \expr(\vec{a}^{m+1}(\vec{x})) \geq \expr(\vec{a}^{m+2}(\vec{x})) \tag{due to \eqref{eq:ev-dec-assumption2} and \eqref{eq:ev-dec-cond-pre}}.
  \end{align*}

  As \eqref{eq:ev-dec-ih} and \eqref{eq:ev-dec-ih2} imply
  \[
    \bigwedge_{m=0}^{n-1}\expr(\vec{a}^m(\vec{x})) \geq \min(\expr(\vec{x}), \expr(\vec{a}^{n-1}(\vec{x}))),
  \]
  this finishes the proof of \eqref{eq:ev-dec-goal} and hence shows \eqref{eq:ev-dec-goal-orig}.

  For exactness, assume $\chi(\vec{x}) \Def \bigwedge_{i=1}^k \expr_i(\vec{x}) > 0$.
  We have
  \begin{align*}
    &\vec{x} \longrightarrow^n_{\langle \chi \land \check\phi, \vec{a}\rangle} \vec{a}^n(\vec{x})\\
    \implies& \chi(\vec{x}) \land \chi(\vec{a}^{n-1}(\vec{x})) \\
    \iff& \bigwedge_{i=1}^k \expr_i(\vec{x}) > 0 \land \expr_i(\vec{a}^{n-1}(\vec{x})) > 0.
  \end{align*}
}

\begin{figure*}
    \begin{align*}
      & \prob{\vec{x}' = \vec{a}^n(\vec{x})}{\top}{x_1 > 0 \land x_3 > 0}{\vec{a}}\\
      \leadsto_e & \prob{\vec{x}' = \vec{a}^n(\vec{x}) \land x_3 > 0}{x_3 > 0}{x_1 > 0}{\vec{a}} \tag{\Cref{thm:conditional-recurrent}} \\
      \leadsto_e & \prob{\vec{x}' = \vec{a}^n(\vec{x}) \land x_3 > 0 \land x_1 > 0 \land x_1^{(n-1)} > 0}{x_3 > 0 \land x_1 > 0}{\top}{\vec{a}} \tag{\Cref{thm:ev-dec-cond}}
    \end{align*}
  \caption{${\leadsto}$-derivation for \Cref{ex:phases-variant}}
  \label{fig:deriv-phases-variant}
\end{figure*}

\begin{example}
  \label{ex:phases-variant}
  Consider the following variant of \ref{eq:phases}
  \[
    \mWhile{x_1 > 0 \land x_3 > 0}{\mat{x_1\\x_2\\x_3} \assign \mat{x_1 + x_2\\x_2 - x_3\\x_3}},
  \]
  i.e., we have $\vec{a} \Def \mat{x_1 + x_2\\x_2 - x_3\\x_3}$.
  Starting with its canonical acceleration problem, we get the derivation shown in \Cref{fig:deriv-phases-variant}, where the second step can be performed via \Cref{thm:ev-dec-cond} as
  \begin{align*}
    & (\check\phi(\vec{x}) \land \expr(\vec{x}) \geq \expr(\vec{a}(\vec{x})))\\
    {} \equiv {} & (x_3 > 0 \land x_1 \geq x_1 + x_2)\\
    {} \equiv {} & (x_3 > 0 \land 0 \geq x_2)
  \end{align*}
  implies
  \begin{align*}
    & (0 \geq x_2 - x_3) \\
    {} \equiv {} & (x_1 + x_2 \geq x_1 + x_2 + x_2 - x_3) \\
    {} \equiv {} & (\expr(\vec{a}(\vec{x})) \geq \expr(\vec{a}^2(\vec{x}))).
  \end{align*}
\end{example}

\subsection{Acceleration via Eventual Increase}

Still, all (combinations of) techniques presented so far fail for
\begin{equation}
  \label{eq:phases2}\tag{\ensuremath{\TT_{ev\text{-}inc}}}
  \mWhile{x_1 > 0}{\mat{x_1\\x_2} \assign \mat{x_1 + x_2\\x_2 + 1}}.
\end{equation}
As in the case of \ref{eq:phases}, the value of $x_1$ does not behave monotonically, i.e., $x_1 > 0$ is neither an invariant nor a converse invariant.
However, this time $x_1$ is eventually \emph{increasing}, i.e., once $x_1$ starts to grow, it never decreases again.
Thus, in this case it suffices to require that $x_1$ is positive and (weakly) increasing.

\begin{theorem}[Acceleration via Eventual Increase]
  \label{thm:ev-inc}
  If $\phi(\vec{x}) \equiv \bigwedge_{i=1}^k C_i$ where each clause $C_i$ contains an inequation $\expr_i(\vec{x}) > 0$ such that
  \[
    \expr_i(\vec{x}) \leq \expr_i(\vec{a}(\vec{x})) \implies \expr_i(\vec{a}(\vec{x})) \leq \expr_i(\vec{a}^2(\vec{x})),
  \]
  then the following acceleration technique is sound:
  \[
    \ref{loop} \mapsto \vec{x}' = \vec{a}^n(\vec{x}) \land \bigwedge_{i=1}^k 0 < \expr_i(\vec{x}) \leq \expr_i(\vec{a}(\vec{x}))
  \]
\end{theorem}
\noindent
We prove the more general \Cref{thm:conditional-ev-inc} later in this section.

\begin{example}
With \Cref{thm:ev-inc}, we can accelerate \ref{eq:phases2} to
\begin{equation}
  \label{psi:phases2}\tag{\ensuremath{\psi_{ev\text{-}inc}}}
  \mat{x_1'\\x_2'} = \mat{\tfrac{n^2-n}{2} + x_2 \cdot n + x_1 \\ x_2 + n} \land 0 < x_1 \leq x_1 + x_2
\end{equation}
as we have
\[
  \begin{aligned}
    & (x_1 \leq x_1 + x_2) \equiv (0 \leq x_2) \implies \\
    & \qquad (0 \leq x_2 + 1) \equiv (x_1 + x_2 \leq x_1 + x_2 + x_2 + 1).
  \end{aligned}
\]
\end{example}
However, \Cref{thm:ev-inc} is \emph{not} exact, as the resulting formula only covers program runs where each $\expr_i$ behaves monotonically.
So \ref{psi:phases2} only covers those runs of \ref{eq:phases2} where the initial value of $x_2$ is non-negative.

Note that \Cref{thm:ev-inc} can also lead to empty under-approximations.
For example, \Cref{thm:ev-inc} can be used to accelerate \ref{loop:exp}, since the implication
\[
  x_1 \leq x_1 - 1 \implies x_1 - 1 \leq x_1 - 2
\]
is valid.
Then the resulting formula is
\[
  \mat{x_1'\\x_2'} = \mat{x_1-n\\2^n \cdot x_2} \land 0 < x_1 \leq x_1 - 1,
\]
which is unsatisfiable.
Thus, when implementing \Cref{thm:ev-inc} (or its conditional version \Cref{thm:conditional-ev-inc}), one has to check whether the resulting formula is satisfiable to avoid trivial (empty) under-approximations.

Again, turning \Cref{thm:ev-inc} into a conditional acceleration technique is straightforward.

\begin{theorem}[Conditional Acceleration via Eventual Increase]
  \label{thm:conditional-ev-inc}
  If we have $\chi(\vec{x}) \equiv \bigwedge_{i=1}^k C_i$ where each clause $C_i$ contains an inequation $\expr_i(\vec{x}) > 0$ such that
  \begin{equation}
    \label{eq:ev-inc-pre}
    \resizebox{0.41\textwidth}{!}{$\check{\phi}(\vec{x}) \land \expr_i(\vec{x}) \leq \expr_i(\vec{a}(\vec{x})) \implies \expr_i(\vec{a}(\vec{x})) \leq \expr_i(\vec{a}^2(\vec{x})),$}
  \end{equation}
  then the following conditional acceleration technique is sound:
  \begin{equation*}
    (\langle \chi, \vec{a} \rangle, \check{\phi}) \mapsto \vec{x}' = \vec{a}^n(\vec{x}) \land \bigwedge_{i=1}^k 0 < \expr_i(\vec{x}) \leq \expr_i(\vec{a}(\vec{x}))
  \end{equation*}
\end{theorem}
\makeproof{thm:conditional-ev-inc}{
  We need to show
  \begin{multline*}
    \vec{x} \longrightarrow_{\langle \check\phi, \vec{a} \rangle}^n \vec{a}^n(\vec{x}) \land \bigwedge_{i=1}^k 0 < \expr_i(\vec{x}) \leq \expr_i(\vec{a}(\vec{x})) \\
    \implies \vec{x} \longrightarrow^n_{\langle \chi, \vec{a} \rangle} \vec{a}^n(\vec{x}).
  \end{multline*}
  Due to $\vec{x} \longrightarrow_{\langle \check\phi, \vec{a} \rangle}^n \vec{a}^n(\vec{x})$, we have
  \begin{equation}
    \label{eq:ev-inc-check}
    \bigwedge_{j=0}^{n-1} \check\phi(\vec{a}^j(\vec{x})).
  \end{equation}
  Let $i$ be arbitrary but fixed and assume
  \begin{equation}
    \label{eq:ev-inc-assumption}
    0 < \expr_i(\vec{x}) \leq \expr_i(\vec{a}(\vec{x})).
  \end{equation}
  We prove
  \begin{equation}
    \label{eq:ev-inc-ih}
    \expr_i(\vec{a}^m(\vec{x})) \leq \expr_i(\vec{a}^{m+1}(\vec{x}))
  \end{equation}
  for all $0 \leq m < n $ by induction on $m$.
  Then we get
  \[
    0 < \expr_i(\vec{a}^m(\vec{x}))
  \]
  and thus $\chi(\vec{a}^m(\vec{x}))$ for all $0 \leq m < n$ due to \eqref{eq:ev-inc-assumption} and hence the claim follows.
  If $m=0$, then
  \[
    \expr_i(\vec{a}^m(\vec{x})) = \expr_i(\vec{x}) \leq \expr_i(\vec{a}(\vec{x})) = \expr_i(\vec{a}^{m+1}(\vec{x})). \tag{due to \eqref{eq:ev-inc-assumption}}
  \]
  In the induction step, note that \eqref{eq:ev-inc-check} implies
  \[
    \check\phi(\vec{a}^m(\vec{x}))
  \]
  as $0 \leq m < n$.
  Together with the induction hypothesis \eqref{eq:ev-inc-ih}, we get
  \[
    \check\phi(\vec{a}^m(\vec{x})) \land \expr_i(\vec{a}^m(\vec{x})) \leq \expr_i(\vec{a}^{m+1}(\vec{x})).
  \]
  By \eqref{eq:ev-inc-pre}, this implies
  \[
    \expr_i(\vec{a}^{m+1}(\vec{x})) \leq \expr_i(\vec{a}^{m+2}(\vec{x})),
  \]
  as desired.
}

\begin{figure*}
    \begin{align*}
      & \prob{\vec{x}' = \vec{a}^n(\vec{x})}{\top}{x_1 > 0 \land x_3 > 0}{\vec{a}}\\
      \leadsto_e & \prob{\vec{x}' = \vec{a}^n(\vec{x}) \land x_3 > 0}{x_3 > 0}{x_1 > 0}{\vec{a}} \tag{\Cref{thm:conditional-recurrent}} \\
      \leadsto_{\phantom{e}} & \prob{\vec{x}' = \vec{a}^n(\vec{x}) \land x_3 > 0 \land 0 < x_1 \leq x_1 + x_2}{x_3 > 0 \land x_1 > 0}{\top}{\vec{a}} \tag{\Cref{thm:conditional-ev-inc}}
    \end{align*}
  \caption{${\leadsto}$-derivation for \Cref{ex:phases2-variant}}
  \label{fig:deriv-phases2-variant}
\end{figure*}

\begin{example}
  \label{ex:phases2-variant}
  Consider the following variant of \ref{eq:phases2}.
  \[
    \mWhile{x_1 > 0 \land x_3 > 0}{\mat{x_1\\x_2\\x_3} \assign \mat{x_1 + x_2\\x_2 + x_3\\x_3}}
  \]
  So we have $\vec{a} \Def \mat{x_1 + x_2\\x_2 + x_3\\x_3}$.
  Starting with the canonical acceleration problem, we get the derivation shown in \Cref{fig:deriv-phases2-variant}, where the second step can be performed via \Cref{thm:conditional-ev-inc} as
  \begin{align*}
    & (\check\phi(\vec{x}) \land \expr(\vec{x}) \leq \expr(\vec{a}(\vec{x}))) \\
    {} \equiv {} & (x_3 > 0 \land x_1 \leq x_1 + x_2) \\
    {} \equiv {} & (x_3 > 0 \land 0 \leq x_2)
  \end{align*}
  implies
  \begin{align*}
    & (0 \leq x_2 + x_3) \\
    {} \equiv {} & (x_1 + x_2 \leq x_1 + x_2 + x_2 + x_3) \\
    {} \equiv {} & (\expr(\vec{a}(\vec{x})) \leq \expr(\vec{a}^2(\vec{x}))).
  \end{align*}
\end{example}

We also considered versions of \Cref{thm:ev-dec-cond,thm:conditional-ev-inc} where the inequations in \eqref{eq:ev-dec-cond-pre} and \eqref{eq:ev-inc-pre} are strict, but this did not lead to an improvement in our experiments.
Moreover, we experimented with a variant of \Cref{thm:conditional-ev-inc} that splits the loop under consideration into two consecutive loops, accelerates them independently, and composes the results.
While such an approach can accelerate loops like \ref{psi:phases2} exactly, the impact on our experimental results was minimal.
Thus, we postpone an in-depth investigation of this idea to future work.

\section{Proving Non-Termination via Loop Acceleration}
\label{sec:nonterm}

We now aim for proving \emph{non-termination}.

\begin{definition}[(Non-)Termination]
  \label{def:term}
  We call a vector $\vec{x} \in \ZZ^d$ a \emph{witness of non-termination} for \ref{loop} (denoted $\nt{\vec{x}}{\ref{loop}}$) if
  \[
    \forall n \in \NN.\ \phi(\vec{a}^n(\vec{x})).
  \]
  If there is such a witness, then \ref{loop} is \emph{non-terminating}.
  Otherwise, \ref{loop} \emph{terminates}.
\end{definition}

To this end, we search for a formula that characterizes a non-empty set of witnesses of non-termination, called a \emph{certificate of non-termination}.

\begin{definition}[Certificate of Non-Termination]
  \label{def:certnonterm}
  We call a formula $\eta \in \Prop{\AAA(\vec{x})}$ a \emph{certificate of non-termination} for \ref{loop} if $\eta$ is satisfiable and
  \[
    \forall \vec{x} \in \ZZ^d.\ \eta(\vec{x}) \implies \nt{\vec{x}}{\ref{loop}}.
  \]
\end{definition}

Clearly, the loops \ref{ex:recurrent} and \ref{eq:phases2} that were used to motivate the acceleration techniques \emph{Acceleration via Monotonic Increase} (\Cref{thm:recurrent}) and \emph{Acceleration via Eventual Increase} (\Cref{thm:ev-inc}) are non-terminating:
\ref{ex:recurrent} diverges for all initial valuations that satisfy its guard $x > 0$ and \ref{eq:phases2} diverges if the initial values are sufficiently large, such that $x_1$ remains positive until $x_2$ becomes non-negative and hence $x_1$ starts to increase.

As we will see in the current section, this is not a coincidence:
Unsurprisingly, all loops that can be accelerated with \emph{Acceleration via Monotonic Increase} or \emph{Acceleration via Eventual Increase} are non-terminating.
More interestingly, the same holds for all loops that can be accelerated using our calculus from \Cref{sec:integration}, as long as all ${\leadsto}$-steps use one of the conditional versions of the aforementioned acceleration techniques, i.e., \emph{Conditional Acceleration via Monotonic Increase} (\Cref{thm:conditional-recurrent}) or \emph{Conditional Acceleration via Eventual Increase} (\Cref{thm:conditional-ev-inc}).
Thus, we obtain a novel, modular technique for proving non-termination of loops \ref{loop}.

Recall that derivations of our calculus from \Cref{sec:integration} start with \emph{canonical acceleration problems} (\Cref{def:acceleration_problem}) whose first component is
\[
  \vec{x}' = \vec{a}^n(\vec{x}).
\]
It relates the values of the variables before evaluating the loop ($\vec{x}$) to the values of the variables after evaluating the loop ($\vec{x}'$) using the closed form ($\vec{a}^n$).
However, if we are interested in non-terminating runs, then the values of the variables after evaluating the loop are obviously irrelevant.
Hence, attempts to prove non-termination operate on a variation of \emph{acceleration problems}, which we call \emph{non-termination problems}.
\begin{definition}[Non-Termination Problem]
  A tuple
  \[
    \ntprob{\psi}{\check{\phi}}{\hat{\phi}}{\vec{a}}
  \]
  where $\psi, \check{\phi}, \hat{\phi} \in \Prop{\AAA(\vec{x})}$ and $\vec{a}: \ZZ^d \to \ZZ^d$ is a \emph{non-termination problem}.
  It is \emph{consistent} if every model of $\psi$ is a witness of non-termination for $\langle \check{\phi}, \vec{a} \rangle$ and \emph{solved} if it is consistent and $\hat{\phi} \equiv \top$.
  The \emph{canonical non-termination problem} of a loop \ref{loop} is
  \[
    \ntprob{\top}{\top}{\phi}{\vec{a}}.
  \]
\end{definition}
In particular, this means that the technique presented in the current section can also be applied to loops where $\vec{a}^n$ cannot be expressed in closed form.

\begin{example}
  \label{ex:canonical-nt}
  The canonical non-termination problem of \ref{eq:phases2} is
  \[
    \ntprob{\top}{\top}{x_1>0}{\mat{x_1+x_2\\x_2+1}}.
  \]
\end{example}

We use a variation of \emph{conditional acceleration techniques} (\Cref{def:cond-accel}), which we call \emph{conditional non-termination techniques}, to simplify the canonical non-termination problem of the analyzed loop.
\begin{definition}[Conditional Non-Termination Technique]
  \label{def:cond-nonterm}
  We call a partial function
  \[
    \nonterm: \MLoop \times \Prop{\AAA(\vec{x})} \rightharpoonup \Prop{\AAA(\vec{x})}
  \]
  a \emph{conditional non-termination technique} if
  \[
    \nt{\vec{x}}{\langle \check{\phi}, \vec{a} \rangle} \land \nonterm(\langle \chi, \vec{a} \rangle,\check{\phi}) \quad \text{implies} \quad \nt{\vec{x}}{\langle \chi, \vec{a} \rangle}
  \]
  for all $(\langle \chi, \vec{a} \rangle,\check{\phi}) \in \dom(\nonterm)$ and all $\vec{x} \in \ZZ^d$.
\end{definition}
Thus, we obtain the following variation of our calculus from \Cref{sec:integration}.
\begin{definition}[Non-Termination Calculus]
  \label{def:nt-calculus}
  The relation ${\leadstont}$ on non-termination problems is defined by the rule
  \[
    \infer{
      \ntprob{\psi_1}{\check{\phi}}{\hat{\phi}}{\vec{a}} \leadstont \ntprob{\psi_1 \cup \psi_2}{\check{\phi} \cup \chi}{\hat{\phi} \setminus \chi}{\vec{a}}
    }{
      \emptyset \neq \chi \subseteq \hat{\phi} & \nonterm(\langle \chi, \vec{a} \rangle, \check{\phi}) = \psi_2
    }
  \]
  where $\nonterm$ is a conditional non-termination technique.
\end{definition}

Like ${\leadsto}$, the relation ${\leadstont}$ preserves consistency.
Hence, we obtain the following theorem, which shows that our calculus is indeed suitable for proving non-termination.

\begin{theorem}[Correctness of ${\leadstont}$]
  \label{thm:nonterm-correct}
  If
  \[
    \ntprob{\top}{\top}{\phi}{\vec{a}} \leadstont^* \ntprob{\psi}{\check{\phi}}{\top}{\vec{a}},
  \]
  and $\psi$ is satisfiable, then $\psi$ is a certificate of non-termination for \ref{loop}.
\end{theorem}
\begin{proof}
  We prove that our calculus preserves consistency, then the claim follows immediately.
  Assume
  \[
    \ntprob{\psi_1}{\check{\phi}}{\hat{\phi}}{\vec{a}} \leadstont \ntprob{\psi_1 \cup \psi_2}{\check{\phi} \cup \chi}{\hat{\phi} \setminus \chi}{\vec{a}}
  \]
  where $\ntprob{\psi_1}{\check{\phi}}{\hat{\phi}}{\vec{a}}$ is consistent and
  \[
    \nonterm(\langle \chi, \vec{a} \rangle, \check{\phi}) = \psi_2.
  \]
  We get
  \begin{align*}
    &\psi_1 \land \psi_2\\
    \implies& \vec{x} \longrightarrow^\infty_{\langle \check{\phi}, \vec{a} \rangle} \bot \land \psi_2 \\
    \implies& \vec{x} \longrightarrow^\infty_{\langle \check{\phi}, \vec{a} \rangle} \bot \land \vec{x} \longrightarrow^\infty_{\langle \chi, \vec{a} \rangle} \bot \\
    \iff& \vec{x} \longrightarrow^\infty_{\langle \check{\phi} \land \chi, \vec{a} \rangle} \bot
  \end{align*}
  The first step holds since $\ntprob{\psi_1}{\check{\phi}}{\hat{\phi}}{\vec{a}}$ is consistent and the second step holds since $\nonterm$ is a conditional non-termination technique.
  This proves consistency of
  \begin{align*}
    &\ntprob{\psi_1 \land \psi_2}{\check{\phi} \land \chi}{\hat{\phi} \setminus \chi}{\vec{a}}\\
    = & \ntprob{\psi_1 \cup \psi_2}{\check{\phi} \cup \chi}{\hat{\phi} \setminus \chi}{\vec{a}}.
  \end{align*}
\end{proof}

Analogously to well-foundedness of ${\leadsto}$, well-foundedness of ${\leadstont}$ is trivial.

\begin{theorem}[Well-Foundedness of ${\leadstont}$]
  \label{thm:nt-term}
  The relation ${\leadstont}$ is well-founded.
\end{theorem}

It remains to present non-termination techniques that can be used with our novel calculus.
We first derive a non-termination technique from \emph{Conditional Acceleration via Monotonic Increase} (\Cref{thm:conditional-recurrent}).

\begin{theorem}[Non-Termination via Monotonic In\-crease]
  \label{thm:nonterm-inc}
  If
  \begin{equation*}
    \check{\phi}(\vec{x}) \land \chi(\vec{x}) \implies \chi(\vec{a}(\vec{x})),
  \end{equation*}
  then
  \[
    (\langle \chi, \vec{a} \rangle, \check{\phi}) \mapsto \chi
  \]
  is a conditional non-termination technique.
\end{theorem}
\begin{proof}
  We need to prove
  \[
    \vec{x} \longrightarrow^\infty_{\langle \check\phi, \vec{a} \rangle} \bot \land \chi(\vec{x}) \implies \vec{x} \longrightarrow^\infty_{\langle \chi, \vec{a} \rangle} \bot.
  \]
  To this end, it suffices to prove
  \[
    \vec{x} \longrightarrow^\infty_{\langle \check\phi, \vec{a} \rangle} \bot \land \chi(\vec{x}) \implies \forall m \in \NN.\ \chi(\vec{a}^m(\vec{x}))
  \]
  by the definition of non-termination (\Cref{def:term}).
  Assume
  \[
    \vec{x} \longrightarrow^\infty_{\langle \check\phi, \vec{a} \rangle} \bot \land \chi(\vec{x}).
  \]
  We prove $\chi(\vec{a}^m(\vec{x}))$ for all $m \in \NN$ by induction on $m$.
  If $m=0$, then the claim follows immediately.
  For the induction step, note that $\vec{x} \longrightarrow^\infty_{\langle \check\phi, \vec{a} \rangle} \bot$ implies $\vec{x} \longrightarrow^{m+1}_{\langle \check\phi, \vec{a} \rangle} \vec{a}^{m+1}(\vec{x})$, which in turn implies $\check\phi(\vec{a}^m(\vec{x}))$.
  Together with the induction hypothesis $\chi(\vec{a}^m(\vec{x}))$, the claim follows from the prerequisites of the theorem.
\end{proof}

\begin{example}
  The canonical non-termination problem of \ref{ex:recurrent} is
  \[
    \ntprob{\top}{\top}{x>0}{(x+1)}.
  \]
  Thus, in order to apply ${\leadstont}$ with \Cref{thm:nonterm-inc}, the only possible choice for the formula $\chi$ is $x>0$.
  Furthermore, we have $\check{\phi} \Def \top$ and $\vec{a} \Def (x+1)$.
  Hence, \Cref{thm:nonterm-inc} is applicable if the implication
  \[
    \top \land x > 0 \implies x+1>0
  \]
  is valid, which is clearly the case.
  Thus, we get
  \[
    \ntprob{\top}{\top}{x>0}{(x+1)} \leadstont \ntprob{x>0}{x>0}{\top}{(x+1)}.
  \]
  Since the latter non-termination problem is solved and $x > 0$ is satisfiable, $x > 0$ is a certificate of non-termination for \ref{ex:recurrent} due to \Cref{thm:nonterm-correct}.
\end{example}

Clearly, \Cref{thm:nonterm-inc} is only applicable in very simple cases.
To prove non-termination of more complex examples, we now derive a conditional non-termination technique from \emph{Conditional Acceleration via Eventual Increase} (\Cref{thm:conditional-ev-inc}).

\begin{theorem}[Non-Termination via Eventual Increase]
  \label{thm:nonterm-ev-inc}
  If we have $\chi(\vec{x}) \equiv \bigwedge_{i=1}^k C_i$ where each clause $C_i$ contains an inequation $\expr_i(\vec{x}) > 0$ such that
  \[
    \check{\phi}(\vec{x}) \land \expr_i(\vec{x}) \leq \expr_i(\vec{a}(\vec{x})) \implies \expr_i(\vec{a}(\vec{x})) \leq \expr_i(\vec{a}^2(\vec{x})),
  \]
  then
  \begin{equation*}
    (\langle \chi, \vec{a} \rangle, \check{\phi}) \mapsto \bigwedge_{i=1}^k 0 < \expr_i(\vec{x}) \leq \expr_i(\vec{a}(\vec{x}))
  \end{equation*}
  is a conditional non-termination technique.
\end{theorem}
\begin{proof}
  Let $\chi' \Def \bigwedge_{i=1}^k 0 < \expr_i(\vec{x}) \leq \expr_i(\vec{a}(\vec{x}))$.
  We need to prove
  \[
    \nt{\vec{x}}{\langle \check{\phi}, \vec{a} \rangle} \land \chi'
    \implies \nt{\vec{x}}{\langle \chi, \vec{a} \rangle}.
  \]
  Then it suffices to prove
  \begin{equation}
    \label{eq:nonterm-ev-inc-to-show}
    \nt{\vec{x}}{\langle \check{\phi},\vec{a} \rangle} \land \chi'(\vec{x}) \implies \nt{\vec{x}}{\langle \chi', \vec{a} \rangle}
  \end{equation}
  since $\chi'$ implies $\chi$.
  By the prerequisites of the theorem, we have $\check{\phi} \land \chi'(\vec{x}) \implies \chi'(\vec{a}(\vec{x}))$.
  Thus, \Cref{thm:nonterm-inc} applies to $\langle \chi', \vec{a} \rangle$.
  Hence, the claim \eqref{eq:nonterm-ev-inc-to-show} follows.
\end{proof}

\begin{example}
  We continue \Cref{ex:canonical-nt}.
  To apply ${\leadstont}$ with \Cref{thm:nonterm-ev-inc} to the canonical non-termination problem of \ref{eq:phases2}, the only possible choice for the formula $\chi$ is $x_1>0$.
  Moreover, we again have $\check{\phi} \Def \top$, and $\vec{a} \Def \mat{x_1+x_2\\x_2+1}$.
  Thus, \Cref{thm:nonterm-ev-inc} is applicable if
  \[
    \top \land x_1 \leq x_1 + x_2 \implies x_1 + x_2 \leq x_1 + 2 \cdot x_2 + 1
  \]
  is valid.
  Since we have $x_1 \leq x_1 + x_2 \iff x_2 \geq 0$ and $x_1 + x_2 \leq x_1 + 2 \cdot x_2 + 1 \iff x_2 + 1 \geq 0$, this is clearly the case.
  Hence, we get
  \begin{align*}
    & \ntprob{\top}{\top}{x_1>0}{\vec{a}} \\
    \leadstont & \ntprob{0 < x_1 \leq x_1 + x_2}{x_1>0}{\top}{\vec{a}}.
  \end{align*}
  Since $0 < x_1 \leq x_1 + x_2 \equiv x_1 > 0 \land x_2 \geq 0$ is satisfiable, $x_1 > 0 \land x_2 \geq 0$ is a certificate of non-termination for \ref{eq:phases2} due to \Cref{thm:nonterm-correct}.
\end{example}

Of course, some non-terminating loops do not behave (eventually) monotonically, as the following example illustrates.

\begin{example}
  \label{ex:fixpoint}
  Consider the loop
  \begin{equation}
    \label{eq:fixpoint}\tag{\ensuremath{\TT_{\textit{fixpoint}}}}
    \mWhile{x_1 > 0}{\mat{x_1\\x_2} \assign \mat{x_2\\x_1}}.
  \end{equation}
  \Cref{thm:nonterm-inc} is inapplicable, since
  \[
    x_1 > 0 \centernot\implies x_2 > 0.
  \]
  Furthermore, \Cref{thm:nonterm-ev-inc} is inapplicable, since
  \[
    x_1 > x_2 \centernot\implies x_2 > x_1.
  \]
\end{example}
  
However, \ref{eq:fixpoint} has \emph{fixpoints}, i.e., there are valuations such that $\vec{x} = \vec{a}(\vec{x})$.
Therefore, it can be handled by existing approaches like \cite[Thm.\ 12]{fmcad19}.
As the following theorem shows, such techniques can also be embedded into our calculus.

\begin{theorem}[Non-Termination via Fixpoints]
  \label{thm:fixpoints}
  For each entity $e$, let $\VV(e)$ denote the set of variables occurring in $e$.
  Moreover, we define
  \[
    \closure_{\vec{a}}(e) \Def \bigcup_{n \in \NN} \VV(\vec{a}^n(e)).
  \]
  Let $\chi(\vec{x}) \equiv \bigwedge_{i=1}^k C_i$, and for each $i \in [1,k]$, assume that $\expr_i(\vec{x}) > 0$ is an inequation that occurs in $C_i$.
  Then
  \begin{equation*}
    (\langle \chi, \vec{a} \rangle, \check{\phi}) \mapsto \bigwedge_{i=1}^k \expr_i(\vec{x}) > 0 \land \bigwedge_{\mathclap{x_j \in \closure_{\vec{a}}(\expr_i)}} x_j = \vec{a}(\vec{x})_j
  \end{equation*}
  is a conditional non-termination technique.
\end{theorem}
\begin{proof}
  Let
  \[
    \chi' \Def \bigwedge_{i=1}^k \expr_i(\vec{x}) > 0 \land \bigwedge_{\mathclap{x_j \in \closure_{\vec{a}}(\expr_i)}} x_j = \vec{a}(\vec{x})_j.
  \]
  We need to prove
  \[
    \nt{\vec{x}}{\langle \check{\phi}, \vec{a} \rangle} \land \chi' \implies \nt{\vec{x}}{\langle \chi, \vec{a} \rangle}.
  \]
  Then it suffices to prove
  \begin{equation}
    \label{eq:fixpoint-to-show}
    \nt{\vec{x}}{\langle \check{\phi}, \vec{a} \rangle} \land \chi'(\vec{x}) \implies \nt{\vec{x}}{\langle \chi', \vec{a} \rangle}
  \end{equation}
  since $\chi'$ implies $\chi$.
  By construction, we have
  \[
    \chi'(\vec{x}) \implies \chi'(\vec{a}(\vec{x})).
  \]
  Thus, \Cref{thm:nonterm-inc} applies to $\langle \chi', \vec{a} \rangle$.
  Hence, the claim \eqref{eq:fixpoint-to-show} follows.
\end{proof}

\begin{example}
  We continue \Cref{ex:fixpoint} by showing how to apply \Cref{thm:fixpoints} to \ref{eq:fixpoint}, i.e., we have $\chi \Def x_1 >0$, $\check{\phi} \Def \top$, and $\vec{a} \Def \mat{x_2\\x_1}$.
  Thus, we get
  \[
    \closure_{\vec{a}}(x_1 > 0) = \{x_1,x_2\}.
  \]
  So starting with the canonical non-termination problem of \ref{eq:fixpoint}, we get
  \begin{align*}
    &\ntprob{\top}{\top}{x_1>0}{\mat{x_2\\x_1}} \\
    \leadstont& \ntprob{x_1 > 0 \land x_1 = x_2}{x_1>0}{\top}{\mat{x_2\\x_1}}.
  \end{align*}
  Since the formula $x_1 > 0 \land x_1 = x_2$ is satisfiable, $x_1 > 0 \land x_1 = x_2$ is a certificate of non-termination for \ref{eq:fixpoint} by \Cref{thm:nonterm-correct}.
\end{example}

Like the acceleration techniques from \Cref{thm:ev-inc,thm:conditional-ev-inc}, the non-termination techniques from \Cref{thm:nonterm-ev-inc,thm:fixpoints} can result in empty under-approximations.
Thus, when integrating these techniques into our calculus, one should check the resulting formula for satisfiability after each step to detect fruitless derivations early.

We conclude this section with a more complex example, which shows how the presented non-termination techniques can be combined to find certificates of non-termination.

\begin{example}
  Consider the following loop:
  \begin{align*}
    & \mathbf{while}\ x_1 > 0 \land x_3 > 0 \land x_4 + 1 > 0\ \mDo\\
    & \qquad \mat{x_1\\x_2\\x_3\\x_4} \assign \mat{1\\x_2 + x_1\\x_3 + x_2\\-x_4}
  \end{align*}
  So we have
  \[
    \phi \Def x_1 > 0 \land x_3 > 0 \land x_4 + 1 > 0
  \]
  and
  \[
    \vec{a} \Def \mat{1\\x_2 + x_1\\x_3 + x_2\\-x_4}.
  \]
  Then the canonical non-termination problem is
  \[
    \ntprob{\top}{\top}{x_1 > 0 \land x_3 > 0 \land x_4 + 1 > 0}{\vec{a}}.
  \]
  First, our implementation applies \Cref{thm:nonterm-inc} to $x_1 > 0$ (as $x_1 > 0 \implies 1 > 0$), resulting in
  \[
    \ntprob{x_1 > 0}{x_1 > 0}{x_3 > 0 \land x_4 + 1 > 0}{\vec{a}}.
  \]
  
  Next, it applies \Cref{thm:nonterm-ev-inc} to $x_3 > 0$, which is possible since
  \[
    x_1 > 0 \land x_3 \leq x_3 + x_2 \implies x_3 + x_2 \leq x_3 + 2 \cdot x_2 + x_1
  \]
  is valid.
  Note that this implication breaks if one removes $x_1 > 0$ from the premise, i.e., \Cref{thm:nonterm-ev-inc} does not apply to $x_3 > 0$ without applying \Cref{thm:nonterm-inc} to $x_1 > 0$ before.
  This shows that our calculus is more powerful than ``the sum'' of the underlying non-termination techniques.
  Hence, we obtain the following non-termination problem:
  \[
  \resizebox{0.48\textwidth}{!}{
    $
      \ntprob{x_1 > 0 \land x_2 \geq 0 \land x_3 > 0}{x_1 > 0 \land x_3 > 0}{x_4 + 1 > 0}{\vec{a}}
    $
  }
  \]
  Here, we simplified
  \[
    0 < x_3 \leq x_3 + x_2
  \]
  to
  \[
    x_2 \geq 0 \land x_3 > 0.
  \]

  Finally, neither \Cref{thm:nonterm-inc} nor \Cref{thm:nonterm-ev-inc} applies to $x_4 + 1 > 0$, since $x_4$ does not behave (eventually) monotonically: Its value after $n$ iterations is given by $(-1)^n \cdot x_4^{\mathit{init}}$, where $x_4^{\mathit{init}}$ denotes the initial value of $x_4$.
  Thus, we apply \Cref{thm:fixpoints} and we get
  \[
    \ntprob{x_1 > 0 \land x_2 \geq 0 \land x_3 > 0 \land x_4 = 0}{\phi}{\top}{\vec{a}},
  \]
  which is solved.
  Here, we simplified the sub-formula $x_4 + 1 > 0 \land x_4 = -x_4$ that results from the last acceleration step to $x_4 = 0$.

  This shows that our calculus allows for applying \Cref{thm:fixpoints} to loops that do not have a fixpoint.
  The reason is that it suffices to require that a \emph{subset} of the program variables remain unchanged, whereas the values of other variables may still change.

  Since
  \[
    x_1 > 0 \land x_2 \geq 0 \land x_3 > 0 \land x_4 = 0
  \]
  is satisfiable, it is a certificate of non-termination due to \Cref{thm:nonterm-correct}.
\end{example}

\section{Related Work}
\label{sec:related}

The related work for our paper splits into papers on acceleration and papers on methods specifically designed to prove non-termination.
In both cases, one major difference between our approach and the approaches in the literature is that we enable a \emph{modular} analysis that allows for combining completely unrelated acceleration or non-termination techniques to process a loop in an iterative way and to reuse information obtained by earlier proof steps.

\subsection{Related Work on Acceleration}
\label{sec:related-accel}

Acceleration-like techniques are also used in \emph{over-ap\-prox\-i\-mat\-ing} settings (see, e.g., \cite{kincaid15,gonnord06,jeannet14,madhukar15,strejcek12,silverman19,kincaid17,schrammel14}), whereas we consider \emph{exact} and \emph{under-ap\-prox\-i\-mat\-ing} loop acceleration techniques.
As many related approaches have already been discussed in \Cref{sec:monotonic}, we only mention three more techniques here.

First, \cite{finite-monid,bozga10} presents an exact acceleration technique for \emph{finite monoid affine transformations} (FMATs), i.e., loops with linear arithmetic whose body is of the form $\vec{x} \assign A\vec{x} + \vec{b}$ where $\{A^i \mid i \in \NN\}$ is finite.
For such loops, Pres\-burger-Arithmetic is sufficient to construct an equivalent formula $\psi$, i.e., it can be expressed in a decidable logic.
In general, this is clearly not the case for the techniques presented in the current paper (which may even synthesize non-polynomial closed forms, see \ref{loop:exp}).
As a consequence and in contrast to our technique, this approach cannot handle loops where the values of variables grow super-linearly (i.e., it cannot handle examples like \ref{eq:three-way-ex}).
Implementations are available in the tools \tool{FAST} \cite{fast} and \tool{Flata} \cite{hojjat12}.
Further theoretical results on linear transformations whose $n$-fold closure is definable in (extensions of) Presburger-Arithmetic can be found in \cite{Boigelot03}.

Second, \cite{bozga09a} shows that \emph{octagonal relations} can be accelerated exactly.
Such relations are defined by a finite conjunction $\xi$ of inequations of the form $\pm x \pm y \leq c$, $x,y \in \vec{x} \cup \vec{x}'$, $c \in \ZZ$.
Then $\xi$ induces the relation $\vec{x} \longrightarrow_\xi \vec{x}' \iff \xi(\vec{x}, \vec{x}')$.
In \cite{octagonsP}, it is proven that such relations can even be accelerated in polynomial time.
This generalizes earlier results for \emph{difference bound constraints} \cite{differenceBounds}.
As in the case of FMATs, the resulting formula can be expressed in Presburger-Arithmetic.
In contrast to the loops considered in the current paper where $\vec{x}'$ is uniquely determined by $\vec{x}$, octagonal relations can represent non-deterministic programs.
Therefore and due to the restricted form of octagonal relations, the work from \cite{bozga09a,octagonsP} is orthogonal to ours.

Third, Albert et al.\ recently presented a technique to find metering functions via \emph{loop specialization}, which is automated via MAX-SMT solving \cite{Albert21}.
This approach could be integrated into our framework via \Cref{thm:conditional-metering}.
However, the technique from \cite{Albert21} focuses on multi-path loops, whereas we focus on single-path loops.
One of the main reasons for our restriction to single-path loops is that their closed form (\Cref{def:closed}) can often be computed automatically in practice.
In contrast, for multi-path loops, it is less clear how to obtain closed forms.

\subsection{Related Work on Proving Non-Termination}
\label{subsec:related-nonterm}

In the following, we focus on approaches for proving non-termination of programs that operate on (unbounded) integer numbers as data.

Many approaches to proving non-termination are based on finding recurrent sets \cite{rupak08}.
A recurrent set describes program configurations from which one can take a step to a configuration that is also in the recurrent set.
Thus, a recurrent set that includes an initial configuration implies non-termination of the program.
In the setting of this paper, a certificate of non-termination $\psi(\vec{x})$ for a loop $\langle \phi, \vec{a} \rangle$ induces the recurrent set
\[
  \{ \vec{a}^n(\vec{x}) \mid n \in \NN \land \vec{x} \in \ZZ^d \land \psi(\vec{x}) \}.
\]
As long as our calculus is used with the non-termination techniques presented in the current paper only (i.e., \Cref{thm:nonterm-inc,thm:nonterm-ev-inc,thm:fixpoints}), it even holds that $\{\vec{x} \in \ZZ^d \mid \psi(\vec{x})\}$ is a recurrent set.
Conversely, a formula characterizing a non-empty recurrent set of a loop is also a certificate of non-termination.
Thus, our calculus could also make use of other non-termination techniques that produce recurrent sets expressed by formulas in $\Prop{\AAA(\vec{x})}$.

Recurrent sets are often synthesized by searching for suitable parameter values for template formulas \cite{rupak08,velroyen,larraz14,anant,geometricNonterm,revterm} or for safety proofs \cite{pldi08,t2-nonterm,seahorn-term}.
In contrast to these search-based techniques, our current techniques use constraint solving only to check implications.
As also indicated by our experiments (\Cref{sec:experiments}), this aspect leads to low runtimes and an efficient analysis.

Many proof techniques for proving non-ter\-mi\-na\-tion of programs \cite{rupak08,jbc-nonterm,anant,geometricNonterm} work by proving that some loop is non-terminating and (often separately) proving that a witness of non-termination for this loop is reachable from an initial program configuration.
This captures \emph{lasso-shaped} counterexamples to program termination.
A \emph{lasso} consists of a \emph{stem} of straight-line code (which could be expressed as a single update), followed by a loop with a single update in its body.
Techniques for proving non-termination of loops that provide witnesses of non-termination can thus be lifted to techniques for lassos by checking reachability of the found witnesses of non-termination from an initial program configuration.
While the presentation in this paper focuses on loops, our implementation in \tool{LoAT} can also prove that the found certificate of non-termination for the loop is reachable from an initial program configuration.
If a loop cannot be proven non-terminating, it can still be possible to accelerate it and use the accelerated loop as part of the stem of a lasso for another loop.
Like this, \tool{LoAT} can analyze programs with more complex control flow than just single loops.

Several techniques for proving \emph{aperiodic} non-termination, i.e., non-termination of programs that do not have any lasso-shaped counterexamples to termination, have been proposed \cite{pldi08,jbc-nonterm,t2-nonterm,larraz14,revterm}.
By integrating our calculus into a suitable program-analysis framework \cite{journal,fmcad19}, it can be used to prove aperiodic non-termination as well.

Loop termination was recently proven to be decidable for the subclass of loops in which the guards and updates use only linear arithmetic and the guards are restricted to conjunctions of atoms \cite{cav19,Hosseini19}.
Our approach is less restrictive regarding the input loops: we allow for non-linear guards and updates, and we allow for arbitrary Boolean structure in the guards.
In future work, one could investigate the use of such decision procedures as conditional non-termination techniques in our calculus to make them applicable to larger classes of loops.
For practical applications to larger programs, it is important to obtain a certificate of non-termination for a loop when proving its non-termination, corresponding to a large, ideally infinite set of witnesses of non-termination.
The reason is that some witness of non-termination for the loop must be reachable from an initial program configuration so that the non-termination proof carries over from the loop to the input program.
However, the decision procedures in \cite{cav19,Hosseini19} are not optimized to this end:
They produce a certificate of \emph{eventual non-termination}, i.e., a formula that describes initial configurations that give rise to witnesses of non-termination by applying the loop body a finite, but unknown number of times.
For example, the most general certificate of non-termination for the loop \ref{ex:recurrent} is $x>0$, whereas the most general certificate of eventual non-termination is $\top$.
The reason is that, for any initial valuation $-c$ of $x$ (where $c$ is a natural number), one obtains a witness of non-termination by applying the body of the loop $c+1$ times while ignoring the loop condition.
The problem of transforming a \emph{single} witness of eventual non-termination into a witness of non-termination has been solved partially in \cite{polyloopsLPAR20}.
The problem of transforming certificates of eventual non-termination that describe \emph{infinite} sets of configurations into certificates of non-termination is, to the best of our knowledge, still open.
In contrast, the conditional non-termination techniques presented in \Cref{sec:nonterm} aim to identify a potentially infinite set of witnesses of non-termination.

For a subclass of loops involving non-linearity and arbitrary Boolean structures, decidability of termination has recently been proven, too \cite{polyloopsSAS20}.
However, the decidability proof from \cite{polyloopsSAS20} only applies to loops where the variables range over $\RR$.
For loops over $\ZZ$, termination of such programs is trivially undecidable (due to Hilbert's tenth problem).

Ben-Amram, Dom{\'{e}}nech, and Genaim \cite{amram19} show a connection between the non-existence of \emph{multiphase-linear ranking functions} as termination arguments for linear loops and \emph{monotonic} recurrent sets.
A recurrent set
\[
  R \Def \left\{\vec{x} \in \ZZ^d \relmiddle{\vert} \bigwedge_{i=1}^m e_i(\vec{x}) > 0\right\}
\]
is monotonic if we have $e_i(\vec{x}) \leq e_i(\vec{a}(\vec{x}))$ for all $i \in [1,m]$ and all $\vec{x} \in R$.
In particular, they propose a procedure that, if it terminates, returns either a multiphase-linear ranking function as a termination proof or a set of program states that could not be proven terminating.
If the input loop has a linear update function with only integer coefficients and if the procedure returns a non-empty set of states that includes integer values, this set is a monotonic recurrent set and proves non-termination.
This approach is implemented in the \tool{iRankFinder} tool.

Leike and Heizmann \cite{geometricNonterm} propose a method to find geometric non-termination arguments that allow for expressing the values of the variables after the $n^{th}$ loop iteration.
In this sense, their approach can also be seen as a use of loop acceleration to express a non-termination argument.
This approach is implemented in the \tool{Ultimate} tool.
While our technique for loop acceleration also relies on closed forms, our technique for proving non-termination does not need closed forms.
Hence our approach also applies to loops where closed forms cannot be computed, or contain sub-expressions that make further analyses difficult, like the imaginary unit.

Finally, Frohn and Giesl \cite{fmcad19} have already used loop acceleration for proving non-termination.
However, they use loop acceleration (more specifically, \Cref{thm:three-way}) solely for proving reachability of non-terminating loops.
To prove non-termination of loops, they used unconditional, standalone versions of \Cref{thm:fixpoints,thm:nonterm-inc}.
So their approach does not allow for combining different acceleration or non-termination techniques when analyzing loops.
However, they already exploited similarities between non-termination proving and loop acceleration:
Both their loop acceleration technique (\Cref{thm:three-way}) and their main non-termination technique (\Cref{thm:nonterm-inc}) are based on certain monotonicity properties of the loop condition.
Starting from this observation, they developed a technique for deducing invariants that may be helpful for both proving non-termination and accelerating loops.
This idea is orthogonal to our approach, which could, of course, be combined with techniques for invariant inference.

\section{Implementation and Experiments}
\label{sec:experiments}

We implemented our approach in our open-source \underline{Lo}op \underline{A}c\-cel\-er\-a\-tion \underline{T}ool \loat \cite{fmcad19,journal}:
\begin{center}
  \url{https://aprove-developers.github.io/LoAT}
\end{center}
It uses \tool{Z3} \cite{z3} and \tool{Yices2} \cite{yices} to check implications and \tool{PURRS} \cite{purrs} to compute closed forms.
While \loat can synthesize formulas with non-polynomial arithmetic, it cannot yet parse them, i.e., the input is restricted to polynomials.
Moreover, \loat does not yet support disjunctive loop conditions.

To evaluate our approach, we extracted $1511$ loops with conjunctive guards from the category \emph{Termination of Integer Transition Systems} of the \emph{Termination Problems Database} \cite{tpdb}, the benchmark collection which is used at the annual \emph{Termination and Complexity Competition} \cite{termcomp}, as follows:
\begin{enumerate}
\item We parsed all examples with \loat and extracted each single-path loop with conjunctive guard (resulting in $3829$ benchmarks).
\item We removed duplicates by checking syntactic equality (resulting in $2825$ benchmarks).
\item We removed loops whose runtime is trivially constant using an incomplete check (resulting in $1733$ benchmarks).
\item We removed loops which do not admit any terminating runs, i.e., loops where \Cref{thm:recurrent} applies (resulting in $1511$ benchmarks).
\end{enumerate}

\noindent
All tests have been run on \tool{StarExec} \cite{starexec} (Intel Xeon E5-2609, 2.40GHz, 264GB RAM \cite{starexec-spec}).
For our benchmark collection, more details about the results of our evaluation, and a pre-compiled binary (Linux, 64 bit) we refer to \cite{website}.

\subsection{Loop Acceleration}
\label{sec:experiments-accel}

For technical reasons, the closed forms computed by \loat are valid only if $n>0$, whereas \Cref{def:closed} requires them to be valid for all $n \in \NN$.
The reason is that \tool{PURRS} has only limited support for initial conditions.
Thus, \loat's results are correct only for all $n>1$ (instead of all $n > 0$).
Moreover, \loat can currently compute closed forms only if the loop body is \emph{triangular}, meaning that each $a_i$ is an expression over $x_1,\ldots,x_i$.
The reason is that \tool{PURRS} cannot solve \emph{systems} of recurrence equations, but only a single recurrence equation at a time.
While systems of recurrence equations can be transformed into a single recurrence equation of higher order, \loat does not yet implement this transformation.
However, \loat failed to compute closed forms for just $26$ out of $1511$ loops in our experiments, i.e., this appears to be a minor restriction in practice.
Furthermore, the implementation of our calculus does not use \emph{conditional acceleration via metering functions}.
The reason is that we are not aware of examples where our monotonicity-based acceleration techniques fail, but our technique for finding metering functions (based on Farkas' Lemma) may succeed.

Apart from these differences, our implementation closely follows the current paper.
It applies the conditional acceleration techniques from \Cref{sec:conditional,sec:accel} with the following priorities: $\text{\Cref{thm:conditional-recurrent}} > \text{\Cref{thm:conditional-one-way}} > \text{\Cref{thm:ev-dec-cond}} > \text{\Cref{thm:conditional-ev-inc}}$.

We compared our implementation with \loat's implementation of \emph{acceleration via monotonicity} (\Cref{thm:three-way}, \cite{fmcad19}) and its implementation of \emph{acceleration via metering functions} (\Cref{thm:meter}, \cite{ijcar16}), which also incorporates the improvements proposed in \cite{journal}.
We did not include the techniques from \Cref{thm:one-way,thm:recurrent} in our evaluation, as they are subsumed by \emph{acceleration via monotonicity}.

Furthermore, we compared with \tool{Flata} \cite{hojjat12}, which implements the techniques to accelerate FMATs and octagonal relations discussed in \Cref{sec:related}.
To this end, we used a straightforward transformation from \loat's native input format\footnote{\url{https://github.com/aprove-developers/LoAT\#koat}} (which is also used in the category \emph{Complexity of Integer Transition Systems} of the \emph{Termination and Complexity Competition}) to \tool{Flata}'s input format.
Note that our benchmark collection contains $16$ loops with non-linear arithmetic where \tool{Flata} is bound to fail, since it supports only linear arithmetic.
We did not compare with \tool{FAST} \cite{fast}, which uses a similar approach as the more recent tool \tool{Flata}.
We used a wall clock timeout of $60$s per example and a memory limit of $128$GB for each tool.

The results can be seen in \Cref{tab1}, where the information regarding the runtime includes all examples, not just solved examples.
They show that our novel calculus was superior to the competing techniques in our experiments.
In all but $7$ cases where our calculus successfully accelerated the given loop, the resulting formula was polynomial.
Thus, integrating our approach into existing acceleration-based verification techniques should not present major obstacles w.r.t.\ automation.

\begin{table}[ht]
\begin{minipage}{0.45\textwidth}
  \begin{center}
    \begin{tabular}{c||c|c|c|c}
                 & \loat    & Monot. & Meter          & \tool{Flata} \\ \hline \hline
      exact      & 1444     & 845    & 0\footnotemark & 1231         \\ \hline
      approx     & 38       & 0      & 733            & 0            \\ \hline
      fail       & 29       & 666    & 778            & 280          \\ \hline
      avg rt     & 0.16s    & 0.11s  & 0.09s          & 0.47s        \\ \hline
      median rt  & 0.09s    & 0.09s  & 0.09s          & 0.40s        \\ \hline
      std dev rt & 0.18s    & 0.09s  & 0.03s          & 0.50s        \\ \hline
    \end{tabular}
    \caption{\loat vs.\ other techniques}
    \label{tab1}
  \end{center}
\end{minipage}
\begin{minipage}{0.45\textwidth}
  \begin{center}
    \begin{tabular}{c||c|c|c}
                 & \sout{Ev-Inc} & \sout{Ev-Dec} & \sout{Ev-Mon} \\ \hline \hline
      exact      & 1444          & 845           & 845           \\ \hline
      approx     & 0             & 493           & 0             \\ \hline
      fail       & 67            & 173           & 666           \\ \hline
      avg rt     & 0.15s         & 0.14s         & 0.09s         \\ \hline
      median rt  & 0.08s         & 0.08s         & 0.07s         \\ \hline
      std dev rt & 0.17s         & 0.17s         & 0.06s         \\ \hline
    \end{tabular}
    \caption{Impact of our new acceleration techniques}
    \label{tab2}
  \end{center}
\end{minipage}

\smallskip
\hfill
\begin{minipage}{0.9\textwidth}
  \setlength{\tabcolsep}{2pt}
  \begin{tabular}{ll}
    \loat:& Acceleration calculus + \Cref{thm:conditional-one-way,thm:conditional-recurrent,thm:ev-dec-cond,thm:conditional-ev-inc}\\
    Monot.:& \emph{Acceleration via Monotonicity}, \Cref{thm:three-way}\\
    Meter:& \emph{Acceleration via Metering Functions}, \Cref{thm:meter}\\
    \tool{Flata}:& \url{http://nts.imag.fr/index.php/Flata}\\
    \sout{Ev-Inc}:& Acceleration calculus + \Cref{thm:conditional-one-way,thm:conditional-recurrent,thm:ev-dec-cond}\\
    \sout{Ev-Dec}:& Acceleration calculus + \Cref{thm:conditional-one-way,thm:conditional-recurrent,thm:conditional-ev-inc}\\
    \sout{Ev-Mon}:& Acceleration calculus + \Cref{thm:conditional-one-way,thm:conditional-recurrent}\\
    exact:& Examples that were accelerated \emph{exactly}\\
    approx:& Examples that were accelerated \emph{approximately}\\
    fail:& Examples that could not be accelerated\\
    avg rt:& Average wall clock runtime\\
    median rt:& Median wall clock runtime\\
    st dev rt:& Standard deviation of wall clock runtime\\
  \end{tabular}
\end{minipage}
\hfill
\end{table}

Furthermore, we evaluated the impact of our new acceleration techniques from \Cref{sec:accel} independently.
To this end, we ran experiments with three configurations where we disabled \emph{acceleration via eventual increase}, \emph{acceleration via eventual decrease}, and both of them.
The results can be seen in \Cref{tab2}.
They show that our calculus does not improve over \emph{acceleration via monotonicity} if both \emph{acceleration via eventual increase} and \emph{acceleration via eventual decrease} are disabled (i.e., our benchmark collection does not contain examples like \ref{eq:beyond-three-way}).
However, enabling either \emph{acceleration via eventual decrease} or \emph{acceleration via eventual increase} resulted in a significant improvement.
Interestingly, there are many examples that can be accelerated with either of these two techniques:
When both of them were enabled, \loat (exactly or approximately) accelerated $1482$ loops.
When only one of them was enabled, it accelerated $1444$ and $1338$ loops, respectively.
But when none of them was enabled, it accelerated only $845$ loops.
We believe that this is due to examples like
\begin{equation}
  \label {eq:ev-mon}
  \mWhile{x_1 > 0 \land \ldots}{\mat{x_1\\x_2\\\ldots} \assign \mat{x_2\\x_2\\\ldots}}
\end{equation}
where \Cref{thm:ev-dec-cond} \emph{and} \Cref{thm:conditional-ev-inc} are applicable (since $x_1 \leq x_2$ implies $x_2 \leq x_2$ and $x_1 \geq x_2$ implies $x_2 \geq x_2$).

\footnotetext{While acceleration via metering functions may be exact in some cases (see the discussion after \Cref{thm:meter}), our implementation cannot check whether this is the case.}

\tool{Flata} exactly accelerated 49 loops where \loat failed or approximated and \loat exactly accelerated 262 loops where \tool{Flata} failed.
So there were only 18 loops where both \tool{Flata} and the full version of our calculus failed to compute an exact result.
Among them were the only 3 examples where our implementation found a closed form, but failed anyway.
One of them was\footnote{The other two are structurally similar, but more complex.}
\[
  \mWhile{x_3 > 0}{\mat{x_1\\x_2\\x_3} \assign \mat{x_1+1\\x_2-x_1\\x_3+x_2}}.
\]
Here, the updated value of $x_1$ depends on $x_1$, the update of $x_2$ depends on $x_1$ and $x_2$, and the update of $x_3$ depends on $x_2$ and $x_3$.
Hence, the closed form of $x_1$ is linear, the closed form of $x_2$ is quadratic, and the closed form of $x_3$ is cubic:
\[
  x_3^{(n)} =-\tfrac{1}{6} \cdot n^3 + \tfrac{1-x_1}{2} \cdot n^2 + \left(\tfrac{x_1}{2} + x_2 - \tfrac{1}{3}\right) \cdot n + x_3
\]
So when $x_1,x_2$, and $x_3$ have been fixed, $x_3^{(n)}$ has up to $2$ extrema, i.e., its monotonicity may change twice.
However, our techniques based on eventual monotonicity require that the respective expressions behave monotonically once they start to de- or increase, so these techniques only allow one change of monotonicity.

This raises the question if our approach can accelerate \emph{every} loop with conjunctive guard and linear arithmetic whose closed form is a vector of (at most) quadratic polynomials with rational coefficients.
We leave this to future work.

\subsection{Non-Termination}
\label{sec:eval-nonterm}

To prove non-termination, our implementation applies the conditional non-termination techniques from \Cref{sec:nonterm} with the following priorities: $\text{\Cref{thm:nonterm-inc}} > \text{\Cref{thm:nonterm-ev-inc}} > \text{\Cref{thm:fixpoints}}$.
To evaluate our approach, we compared it with several leading tools for proving non-termination of integer programs: \tool{AProVE} \cite{tool-jar}, \tool{iRankFinder} \cite{amram19}, \tool{RevTerm} \cite{revterm}, \tool{Ultimate} \cite{Ultimate}, and \tool{VeryMax} \cite{larraz14}.
Note that \tool{AProVE} uses, among other techniques, the tool \tool{T2} \cite{t2-tool} as backend for proving non-termination, so we refrained from including \tool{T2} separately in our evaluation.

To compare with \tool{AProVE}, \tool{RevTerm}, and \tool{Ultimate}, we transformed all examples into the format which is used in the category \emph{Termination of C Integer Programs} of the \emph{Termination and Complexity Competition}.\footnote{\url{http://termination-portal.org/wiki/C_Integer_Programs}}
For \tool{iRankFinder} and \tool{VeryMax}, we transformed them into the format from the category \emph{Termination of Integer Transition Systems} of the \emph{Termination and Complexity Competition} \cite{smtlib-format}.
The latter format is also supported by \tool{LoAT}, so besides \tool{iRankFinder} and \tool{VeryMax}, we also used it to evaluate \tool{LoAT}.

For the tools \tool{iRankFinder}, \tool{Ultimate}, and \tool{VeryMax}, we used the versions of their last participations in the \emph{Termination and Complexity Competition} (2019 for \tool{VeryMax} and 2021 for \tool{iRankFinder} and \tool{Ultimate}), as suggested by the developers.
For \tool{AProVE}, the developers provided an up-to-date version.
For \tool{RevTerm}, we followed the build instruction from \cite{revterm-website} and sequentially executed the following command lines, as suggested by the developers:

\smallskip
\noindent
\resizebox{0.48\textwidth}{!}{
\tt
RevTerm.sh prog.c -linear part1 mathsat 2 1
}
\noindent
\resizebox{0.48\textwidth}{!}{
\tt
RevTerm.sh prog.c -linear part2 mathsat 2 1
}

It is important to note that all tools but \tool{RevTerm} and \loat also try to prove termination.
Thus, a comparison between the runtimes of \loat and those other tools is of limited significance.
Therefore, we also compared \loat with configurations of the tools that only try to prove non-termination.
For \tool{AProVE}, such a configuration was kindly provided by its developers (named \tool{AProVE NT} in \Cref{tab3}).
In the case of \tool{iRankFinder}, the excellent documentation allowed us to easily build such a configuration ourselves (named \tool{iRankFinder NT} in \Cref{tab3}).
In the case of \tool{Ultimate}, its developers pointed out that a comparison w.r.t.\ runtime is meaningless, as it is dominated by \tool{Ultimate}'s startup-time of ${\sim} 10$s on small examples.
For \tool{VeryMax}, it is not possible to disable termination-proving, according to its authors.

\smallskip
\noindent
We again used a wall clock timeout of $60$s and a memory limit of $128$GB for each tool.
For \tool{RevTerm}, we used a timeout of $30$s for the first invocation, and the remaining time for the second invocation.

\begin{table*}[ht]
  \begin{center}
    \scalebox{0.85}{
    \begin{tabular}{c||c|c|c|c|c|c|c|c|c|}
                       & \loat   & \tool{AProVE} & \tool{AProVE NT} & \tool{iRankFinder} & \tool{iRankFinder NT} & \tool{RevTerm} & \tool{Ultimate} & \tool{VeryMax} \\ \hline \hline
      NO               & 206     & 200           & 200              & 205                & 205                   & 133            & 205             & 175            \\ \hline
      YES              & 0       & 1301          & 0                & 1298               & 0                     & 0              & 919             & 1299           \\ \hline
      fail             & 1305    & 10            & 1311             & 8                  & 1306                  & 1378           & 387             & 37             \\ \hline
      avg rt           & 0.03s   & 16.09s        & 25.69s           & 1.40s              & 1.03s                 & 38.85s         & 23.30s          & 3.17s          \\ \hline
      avg rt NO        & 0.03s   & 10.65s        & 9.67s            & 1.34s              & 0.96s                 & 4.63s          & 7.99s           & 14.52s         \\ \hline
      median rt        & 0.02s   & 13.41s        & 15.74s           & 1.11s              & 0.91s                 & 34.93s         & 11.57s          & 0.03           \\ \hline
      median rt NO     & 0.02s   & 8.51s         & 6.91s            & 1.17s              & 0.90s                 & 1.88s          & 8.01s           & 5.36s           \\ \hline
      std dev rt       & 0.03s   & 10.09s        & 19.95s           & 2.25s              & 0.30s                 & 16.61s         & 21.60s          & 10.76s         \\ \hline
      std dev rt NO    & 0.06s   & 5.80s         & 5.80s            & 0.92s              & 0.12s                 & 11.08s         & 1.88s           & 13.24s         \\ \hline
    \end{tabular}
    }
    \caption{Comparison of \loat with competing tools}
    \label{tab3}
    \smallskip
    \scalebox{0.85}{
    \begin{tabular}{c||c|c|c|c|c|c|c|c|c|}
                       & \loat   & \sout{Inc} & \sout{Ev-Inc} & \sout{FP} & \sout{Ev-Inc, FP} \\ \hline \hline
      NO               & 206     & 206        & 203           & 205       & 0                 \\ \hline
      fail             & 1305    & 1305       & 1308          & 1306      & 1511              \\ \hline
      avg rt           & 0.03s   & 0.03s      & 0.03s         & 0.03s     & 0.02s             \\ \hline
      avg rt NO        & 0.03s   & 0.02s      & 0.02s         & 0.02s     & --                \\ \hline
      median rt        & 0.02s   & 0.02s      & 0.02s         & 0.02s     & 0.02s             \\ \hline
      median rt NO     & 0.02s   & 0.02s      & 0.02s         & 0.02s     & --                \\ \hline
      std dev rt       & 0.03s   & 0.02s      & 0.02s         & 0.02s     & 0.02s             \\ \hline
      std dev rt NO    & 0.06s   & 0.03s      & 0.03s         & 0.04s     & --                \\ \hline
    \end{tabular}
    }
    \caption{Comparison of \loat versions}
    \label{tab4}
    \smallskip
    \begin{tabular}{ll}
      \loat:& Calculus from \Cref{sec:nonterm}\\
      \tool{AProVE}:& \url{https://aprove.informatik.rwth-aachen.de}\\
      \tool{AProVE NT}:& Configuration of \tool{AProVE} that does not try to prove termination \\
      \tool{iRankFinder}:& \url{http://irankfinder.loopkiller.com}\\
      \tool{iRankFinder NT}:& Configuration of \tool{iRankFinder} that does not try to prove termination \\
      \tool{RevTerm}:& \url{https://github.com/ekgma/RevTerm} \\
      \tool{Ultimate}:& \url{https://monteverdi.informatik.uni-freiburg.de/tomcat/Website} \\
      \tool{VeryMax}:& \url{https://www.cs.upc.edu/~albert/VeryMax.html}\\
      \sout{Inc}:& Calculus from \Cref{sec:nonterm} without \Cref{thm:nonterm-inc}\\
      \sout{Ev-Inc}:& Calculus from \Cref{sec:nonterm} without \Cref{thm:nonterm-ev-inc}\\
      \sout{FP}:& Calculus from \Cref{sec:nonterm} without \Cref{thm:fixpoints}\\
      \sout{Ev-Inc, FP}:& Calculus from \Cref{sec:nonterm} without \Cref{thm:nonterm-ev-inc,thm:fixpoints}\\
      NO:& Number of non-termination proofs\\
      YES:& Number of termination proofs\\
      fail:& Number of examples where (non-)termination could not be proven\\
      avg rt:& Average wall clock runtime\\
      avg rt NO:& Average wall clock runtime when non-termination was proven\\
      median rt:& Median wall clock runtime\\
      median rt NO:& Median wall clock runtime when non-termination was proven\\
      st dev rt:& Standard deviation of wall clock runtime\\
      st dev rt NO:& Standard deviation of wall clock runtime when non-termination was proven\\
    \end{tabular}
  \end{center}
\end{table*}

The results can be seen in \Cref{tab3}.
They show that our novel calculus is competitive with state-of-the-art tools.
Both \tool{iRankFinder} and \tool{Ultimate} can prove non-termination of precisely the same $205$ examples.
\loat can prove non-termination of these examples, too.
In addition, it solves one benchmark that cannot be handled by any other tool:\footnote{\tt 1567523105272726.koat.smt2}

\[
  \mWhile{x>9 \land x_1 \geq 0}{\mat{x\\x_1} \assign \mat{x_1^2+2 \cdot x_1 + 1\\x_1+1}}
\]

Most likely, the other tools fail for this example due to the presence of non-linear arithmetic.
Our calculus from \Cref{sec:nonterm} just needs to check implications, so as long as the underlying SMT-solver supports non-linearity, it can be applied to non-linear examples, too.
It is worth mentioning that \loat subsumes all other tools w.r.t.\ proving non-termination.
There are only $4$ examples where none of the tools can prove termination or non-termination.
Termination of one of these examples can be proven by an experimental, unpublished module of \loat, which is inspired by the calculi presented in this paper.
A manual investigation revealed that the $3$ remaining examples are terminating, too.

To investigate the impact of the different non-termination techniques, we also tested configurations where one of the non-termination techniques from \Cref{thm:nonterm-inc,thm:nonterm-ev-inc,thm:fixpoints} was disabled, respectively.
The results can be seen in \Cref{tab4}.
First, note that disabling \Cref{thm:nonterm-inc} does not reduce \loat's power.
The reason is that if the left-hand side $p$ of an inequation $p>0$ is monotonically increasing (such that \Cref{thm:nonterm-inc} applies), then it is also eventually monotonically increasing (such that \Cref{thm:nonterm-ev-inc} applies).
However, since \Cref{thm:nonterm-inc} yields simpler certificates of non-termination than \Cref{thm:nonterm-ev-inc}, \loat still uses both techniques.
Interestingly, \sout{Ev-Inc} and \sout{FP} are almost equally powerful:
Without \Cref{thm:nonterm-ev-inc}, \loat still proves non-termination of 203 examples and without \Cref{thm:fixpoints}, \loat solves 205 examples.
Presumably, the reason is again due to examples like \eqref{eq:ev-mon}, where \Cref{thm:nonterm-ev-inc} finds the recurrent set $x_2 > 0$ and \Cref{thm:fixpoints} finds the recurrent set $x_1 > 0 \land x_1 = x_2$.
So even though both non-termination techniques are applicable in such cases, the recurrent set deduced via \Cref{thm:nonterm-ev-inc} is clearly more general and thus preferable in practice.
Note that \loat cannot solve a single example when \emph{both} \Cref{thm:nonterm-ev-inc} and \Cref{thm:fixpoints} are disabled (\sout{Ev-Inc, FP} in \Cref{tab4}).
Then the only remaining non-termination technique is \emph{Non-Termination via Monotonic Increase}.
Examples where this technique suffices to prove non-termination trivially diverge whenever the loop condition is satisfied, and hence they were filtered from our benchmark set (as explained at the beginning of \Cref{sec:experiments}).

Regarding the runtime, we think that \tool{LoAT} is faster than the competing tools due to the fact that the technique presented in \Cref{sec:nonterm} requires very little search, whereas many other non-termination techniques are heavily search-based (e.g., due to the use of \emph{templates}, as it is exercised by \tool{RevTerm}).
In our setting, the inequations that eventually constitute a certificate of non-termination immediately arise from the given loop.
In this regard, \tool{iRankFinder}'s approach for proving non-termination is similar to ours, as it also requires little search.
This is also reflected in our experiments, where \tool{iRankFinder} is the second fastest tool.

It should also be taken into account that \tool{iRankFinder} is implemented in {\sf Python}, \tool{AProVE}, \tool{RevTerm}, and \tool{Ultimate} are implemented in {\sf Java}, and \loat and \tool{VeryMax} are implemented in \CC.
Thus, the difference in runtime is in parts due to the different performances of the respective programming language implementations.

Another interesting aspect of our evaluation is the result of \tool{RevTerm}, which outperformed all other tools in the evaluation of \cite{revterm}.
The reason for this discrepancy is the following:
In \cite{revterm}, $300$ different configurations of \tool{RevTerm} have been tested, and a benchmark has been considered to be solved if at least one of these configurations was able to prove non-termination.
In contrast, we ran two configurations of \tool{RevTerm}, one for each of the two non-termination checks proposed in \cite{revterm}.
So essentially, \tool{RevTerm}'s results from \cite{revterm} correspond to a highly parallel setting, whereas the results from our evaluation correspond to a sequential setting.

\section{Conclusion and Future Work}
\label{sec:conclusion}

We discussed existing acceleration techniques (\Cref{sec:monotonic}) and presented a calculus to combine acceleration techniques modularly (\Cref{sec:integration}).
Then we showed how to combine existing (\Cref{sec:conditional}) and two novel (\Cref{sec:accel}) acceleration techniques with our calculus.
This improves over prior approaches, where acceleration techniques were used independently, and may thus improve acceleration-based verification techniques \cite{underapprox15,fmcad19,journal,bozga09a,bozga10,iosif17} in the future.
An empirical evaluation (\Cref{sec:experiments-accel}) shows that our approach is more powerful than state-of-the-art acceleration techniques.
Moreover, if it is able to accelerate a loop, then the result is exact (instead of just an under-approximation) in most cases.
Thus, our calculus can be used for under-approximating techniques (e.g., to find bugs or counterexamples) as well as in over-approximating settings (e.g., to prove safety or termination).

Furthermore, we showed how our calculus from \Cref{sec:integration} can be adapted for proving non-termination in \Cref{sec:nonterm}, where we also presented three non-termination techniques that can be combined with our novel calculus.
While two of them (\Cref{thm:nonterm-inc,thm:fixpoints}) are straightforward adaptions of existing non-termination techniques to our modular setting, the third one (\Cref{thm:nonterm-ev-inc}) is, to the best of our knowledge, new and might also be of interest independently from our calculus.

Actually, the two calculi presented in this paper are so similar that they do not require separate implementations.
In our tool \tool{LoAT}, both of them are implemented together, such that we can handle loops uniformly:
If our implementation of the calculi yields a certificate of non-termination, then it suffices to prove reachability of one of the corresponding witnesses of non-termination from an initial program state afterwards to finish the proof of non-termination.
If our implementation of the calculi successfully accelerates the loop under consideration, this may help to prove reachability of other, potentially non-terminating configurations later on.
If our implementation of the calculi fails, then \tool{LoAT} continues its search with other program paths.
The success of this strategy is demonstrated at the annual \emph{Termination and Complexity Competition}, where \tool{LoAT} has been the most powerful tool for proving non-termination of \emph{Integer Transition Systems} since its first participation in 2020.

Regarding future work, we are actively working on support for disjunctive loop conditions.
Moreover, our experiments indicate that integrating specialized techniques for FMATs (see \Cref{sec:related}) would improve the power of our approach for loop acceleration, as \tool{Flata} exactly accelerated 49 loops where \tool{LoAT} failed to do so (see \Cref{sec:experiments}).
Furthermore, we plan to extend our approach to richer classes of loops, e.g., loops operating on both integers and arrays, non-deterministic loops, or loops operating on bitvectors (as opposed to mathematical integers).

\subsubsection*{Acknowledgments}

We thank Marcel Hark, Sophie Tourret, and the anonymous reviewers for helpful feedback and comments.
Moreover, we thank Jera Hensel for her help with \tool{AProVE}, Radu Iosif and Filip Kone\v{c}n\'{y} for their help with \tool{Flata}, Samir Genaim for his help with \tool{iRankFinder}, Matthias Heizmann for his help with \tool{Ultimate}, Ehsan Goharshady for his help with \tool{RevTerm}, and Albert Rubio for his help with \tool{VeryMax}.

This work has been funded by the Deutsche Forschungsgemeinschaft (DFG, German Research Foundation) - 235950644 (Project GI 274/6-2),
    and by the Deutsche Forschungsgemeinschaft (DFG, German Research Foundation) - 389792660 as part of TRR 248 (see \url{https://perspicuous-computing.science}).

\bibliographystyle{splncs04}
\bibliography{refs,crossrefs,strings}

\end{document}